\documentclass[12pt,a4paper,english]{article}
\usepackage[ansinew]{inputenc}
\usepackage[T1]{fontenc}
\usepackage[english]{babel}
\usepackage{tabularx}
\usepackage{fancybox}
\usepackage{fancyhdr}
\usepackage{color}
\usepackage{authblk}
\usepackage{setspace}
\usepackage[numbers, sort]{natbib}

\bibpunct{[}{]}{,}{n}{,}{,}
\usepackage{amsmath}
\usepackage{amstext}
\usepackage{amssymb}
\usepackage{amsthm}
\usepackage{mathrsfs}
\newcommand{\be}{\begin{equation}}
\newcommand{\ee}{\end{equation}}
\newcommand{\ba}{\begin{aligned}}
\newcommand{\ea}{\end{aligned}}
\newcommand{\R}{\mathbb{R}}
\newcommand{\N}{\mathbb{N}}
\def\spazio{\kern-.09em }
\newcommand{\ind}{\mathbf{1}}
\newcommand{\lsi}{\left[\negthinspace\left[}
\newcommand{\rsi}{\right]\negthinspace\right]}
\newcommand{\lsir}{\left(\spazio\negthinspace\left(}
\newcommand{\rsir}{\right)\spazio\negthinspace\right)}
\newcommand{\M}{\mathcal{M}}
\newcommand{\Mloc}{\mathcal{M}_{\text{loc}}}
\newcommand{\mMloc}{\mathcal{M}_{\emph{loc}}}
\newcommand{\D}{\mathcal{D}}
\newcommand{\Dweak}{\mathcal{D}^{\text{weak}}}
\newcommand{\mDweak}{\mathcal{D}^{\emph{weak}}}
\newcommand{\A}{\mathcal{A}}
\newcommand{\F}{\mathcal{F}}
\newcommand{\FF}{\mathbb{F}}

\newcommand{\hatK}{\widehat{K}}
\newcommand{\hatZ}{\widehat{Z}}
\newcommand{\hatQ}{\widehat{Q}}
\newcommand{\hatV}{\widehat{V}}
\newcommand{\tildeS}{\widetilde{S}}

\newcommand{\tildeH}{\widetilde{H}}
\newcommand{\tildeK}{\widetilde{K}}
\newcommand{\tildeC}{\widetilde{C}}
\newcommand{\bit}{\bibitem}
\DeclareMathOperator{\var}{Var}
\DeclareMathOperator{\sign}{sign}
\DeclareMathOperator{\conv}{conv}

\newtheorem{Thm}{\bf Theorem}[section]
\newtheorem{Def}[Thm]{\bf Definition}
\newtheorem{Prop}[Thm]{\bf Proposition}
\newtheorem{Lem}[Thm]{\bf Lemma}
\newtheorem{Cor}[Thm]{\bf Corollary}
\theoremstyle{remark}
\newtheorem{Rem}[Thm]{\bf Remark}
\newtheorem{Rems}[Thm]{\bf Remarks}
\newtheorem{Ex}[Thm]{\bf Example}

\numberwithin{equation}{section}
\usepackage{geometry}
\makeatletter
\renewcommand*\@fnsymbol[1]{\the#1}
\makeatother

\title{Weak and strong no-arbitrage conditions \\ for continuous financial markets}

\author{\vspace{0.2cm} Claudio Fontana}
\affil{\normalsize{
Laboratoire de Math\'ematiques et Mod\'elisation, Universit\'e d'\'Evry Val d'Essonne \\\vspace{-0.1cm} 
23 Boulevard de France, \'Evry Cedex, 91037, France.} \\\vspace{0.2cm}
E-mail: \texttt{claudio.fontana@univ-evry.fr}}

\date{This version: May 13, 2014}

\begin{document}

\maketitle

\abstract{\begin{spacing}{1.1}\noindent We propose a unified analysis of a whole spectrum of no-arbitrage conditions for financial market models based on continuous semimartingales. In particular, we focus on no-arbitrage conditions weaker than the classical notions of \emph{No Arbitrage} and \emph{No Free Lunch with Vanishing Risk}. We provide a complete characterisation of the considered no-arbitrage conditions, linking their validity to the characteristics of the discounted asset price process and to the existence and the properties of \emph{(weak) martingale deflators}, and review classical as well as recent results.\end{spacing}}
\vspace{0.5cm}

\begin{spacing}{1}\noindent \textbf{Keywords:} arbitrage, benchmark approach, continuous semimartingale, martingale deflator, market price of risk, arbitrage of the first kind, free lunch with vanishing risk.\\[-0.2cm]

\end{spacing}

\section{Introduction}	\label{S1}

Modern mathematical finance is strongly rooted on the \emph{no-arbitrage} paradigm. In a nutshell, this amounts to excluding the possibility of ``making money out of nothing'' by cleverly trading in the financial market. Since the existence of such a possibility is both unrealistic and, loosely speaking, conflicts with the existence of an economic equilibrium, any mathematical model for a realistic financial market is required to satisfy a suitable no-arbitrage condition, in the absence of which one cannot draw meaningful conclusions on asset prices and investors' behavior.

The search for a satisfactory no-arbitrage condition has a rather long history, grown at the border between financial economics and mathematics. We do not attempt here a detailed overview of the historical developments of modern no-arbitrage theory, but only mention the seminal papers \cite{HK}-\cite{HP} and refer the reader to \cite{DS} and \cite{Sch} for more information. A decisive step in this history was marked by the paper \cite{DS94}, where, in the case of locally bounded processes, it was proved the equivalence between the \emph{No Free Lunch with Vanishing Risk (NFLVR)} condition (a condition slightly stronger than the classical \emph{No Arbitrage (NA)} condition) and the existence of an \emph{Equivalent Local Martingale Measure (ELMM)}, i.e., a new probability measure equivalent to the original one such that the discounted asset price process is a local martingale under the new measure. The local boundedness assumption was then removed in the subsequent papers \cite{DS98b} and \cite{Kab}.

The NFLVR condition has established itself as a golden standard and the vast majority of models proposed in quantitative finance satisfies it. However, financial market models that fail to satisfy NFLVR have also appeared in recent years. In particular, in the context of the \emph{Benchmark Approach} (see \cite{Pl1}-\cite{PH}), a new asset pricing theory has been developed without relying on the existence of ELMMs. A similar perspective is also adopted in the \emph{Stochastic Portfolio Theory} (see \cite{Fern}-\cite{FK}), where the NFLVR condition is not imposed as a normative assumption and it is shown that arbitrage opportunities may naturally arise in the market.
Related works that explicitly consider situations where NFLVR may fail are \cite{Cass}, \cite{CT}, \cite{CL}, \cite{Font}, \cite{FJS}, \cite{HS}, \cite{KK}, \cite{Ka1}, \cite{MR}, \cite{RR} and also, in the more specific case of diffusion models, \cite{FR}, \cite{LW}, \cite{Lo} and \cite{Ruf} (see later in the text for more information). Somewhat surprisingly, these works have shown that the full strength of NFLVR is not necessarily needed in order to solve the fundamental problems of valuation, hedging and portfolio optimisation. However, the situation is made complicated by the fact that many different and alternative no-arbitrage conditions have been proposed in the literature during the last two decades.

Motivated by the preceding discussion, the present paper aims at presenting a unified and clear perspective on the most significant no-arbitrage conditions in the context of general financial market models based on continuous semimartingales\footnote{The continuous semimartingale setting allows for a rather transparent analysis and covers many models widely used in quantitative finance (in particular, almost all models developed in the context of Stochastic Portfolio Theory).}. In particular, we study several no-arbitrage conditions which are \emph{weaker} than the classical and \emph{strong} NFLVR condition, namely the \emph{No Increasing Profit (NIP)}, \emph{No Strong Arbitrage (NSA)} and \emph{No Arbitrage of the First Kind (NA1)} conditions. We prove the following implications:
\be	\label{impl}
\text{ NFLVR $\Longrightarrow$ NA1 $\Longrightarrow$ NSA $\Longrightarrow$ NIP. }
\ee
By means of explicit examples and counterexamples, we illustrate these implications and discuss their economic meaning, their relations with the Benchmark Approach as well as the connections to several other conditions which have appeared in the literature, thus providing a complete picture of a whole spectrum of no-arbitrage conditions. Moreover, we prove that none of the converse implications in \eqref{impl} holds in general.

We show that \emph{weak} no-arbitrage conditions (i.e., NIP, NSA and NA1) can be completely characterised in terms of the semimartingale characteristics of the discounted price process, while this is in general impossible for \emph{strong} no-arbitrage conditions (NA and NFLVR), since the latter also depend on the structure of the underlying filtration. Moreover, we link the validity of different no-arbitrage conditions to the existence and the properties of \emph{(weak) martingale deflators}, which can be thought of as weaker counterparts of density processes of ELMMs. In particular, we show that the weak NSA and NA1 conditions (as well as their equivalent formulations) can be directly checked by looking at the \emph{minimal} weak martingale deflator, the properties of which are easily determined by the mean-variance trade-off process of the discounted price process. Furthermore, we prove that NA1 (as well as its equivalent formulations) is stable with respect to changes of numéraire (see Corollary \ref{stable}), unlike the classical NFLVR condition, and allows to recover NFLVR by means of a suitable change of numéraire (see Corollary \ref{NA1-NFLVR}).

Altogether, referring to Section \ref{S7} for a more detailed discussion on the economic aspects, the present study shows that the NIP and NSA conditions can be regarded as indispensable no-arbitrage requirements for any realistic financial market, but are not enough for the purposes of financial modeling. On the contrary, the NA1 condition, even though strictly weaker than the classical NFLVR condition, is equivalent to a meaningful notion of \emph{market viability} and allows for a satisfactory solution of all typical problems of mathematical finance.

To the best of our knowledge, there does not exist in the literature a similar unifying analysis of the weak no-arbitrage conditions going beyond the classical notions of NA and NFLVR. The only exception is contained in Chapter 1 of \cite{Hulley}. In comparison with the latter work, our approach puts more emphasis on the role of (weak) martingale deflators and also carefully takes into account minimal no-arbitrage conditions that are weaker than the NUPBR condition, on which the presentation in \cite{Hulley} is focused. Moreover, besides providing different and original proofs, we also study the NIP, NA1, NCT and NAA no-arbitrage conditions (see e.g. the table in Section \ref{S7}), which are not explicitly considered in \cite{Hulley}, and drop the non-negativity assumption on the discounted asset prices.

The paper is structured as follows. Section \ref{S2} presents the general setting and introduces the main no-arbitrage conditions which shall be studied in the following. Section \ref{S3}, \ref{S4} and \ref{S5} discuss the NIP, NSA and NA1 conditions, respectively. Section \ref{S6} deals with the classical NFLVR condition and discusses its relations with the previous no-arbitrage conditions. Section \ref{examples} illustrates the implications \eqref{impl} by means of several examples. Finally, Section \ref{S7} concludes by summarising the different no-arbitrage conditions studied in the present paper and commenting on their economic implications.

\section{General setting and preliminaries}	\label{S2}

Let $(\Omega,\F,\FF,P)$ be a given filtered probability space, where the filtration $\FF=(\F_t)_{0\leq t \leq T}$ is assumed to satisfy the usual conditions of right-continuity and $P$-completeness and, for the sake of simplicity, $T\in\left(0,\infty\right)$ represents a finite time horizon (all the results we are going to present can be rather easily adapted to the infinite horizon case). Note that the initial $\sigma$-field $\F_0$ is not assumed to be trivial. Let $\M$ denote the family of all uniformly integrable $P$-martingales and $\Mloc$ the family of all local $P$-martingales. Without loss of generality, we assume that all processes in $\Mloc$ have càdlàg paths and we denote by $\M^c$ and $\Mloc^c$ the families of all processes in $\M$ and $\Mloc$, respectively, with continuous paths.

We consider a financial market comprising $d+1$ assets, whose prices are represented by the $\R^{d+1}$-valued process $\tildeS=(\tildeS_t)_{0\leq t \leq T}$, with $\tildeS_t=(\tildeS^0_t,\tildeS^1_t,\ldots,\tildeS^d_t)^{\top}$, with $^{\top}$ denoting transposition. We assume that $\tildeS^0_t$ is $P$-a.s. strictly positive for all $t\in[0,T]$ and, as usual in the literature, we then take asset $0$ as numéraire and express all quantities in terms of $\tildeS^0$. This means that the ($\tildeS^0$-discounted) price of asset $0$ is constant and equal to $1$ and the remaining $d$ risky assets have ($\tildeS^0$-discounted) prices described by the $\R^d$-valued process $S=(S_t)_{0\leq t \leq T}$, with $S^i_t:=\tildeS^i_t/\tildeS^0_t$ for all $t\in[0,T]$ and $i=1,\ldots,d$. The process $S$ is assumed to be a \emph{continuous} $\R^d$-valued semimartingale on $(\Omega,\F,\FF,P)$. In particular, $S$ is a special semimartingale, admitting the unique canonical decomposition $S=S_0+A+M$, where $A$ is a continuous $\R^d$-valued predictable process of finite variation and $M$ is an $\R^d$-valued process in $\Mloc^c$ with $M_0=A_0=0$. Due to Proposition II.2.9 of \cite{JS}, it holds that, for all $i,j=1,\ldots,d$,
\be	\label{char}
A^i = \int\!a^idB
\qquad\text{and}\qquad
\langle S^i,S^j\rangle = \langle M^i,M^j\rangle = \int\!c^{ij} dB,
\ee
where $B$ is a continuous real-valued predictable increasing process, $a=(a^1,\ldots,a^d)^{\top}$ is an $\R^d$-valued predictable process and $c=\bigl((c^{i1})_{1\leq i \leq d},\ldots,(c^{id})_{1\leq i \leq d}\bigr)$ is a predictable process taking values in the cone of symmetric nonnegative $d\times d$ matrices. The processes $a$, $c$ and $B$ in \eqref{char} are not unique in general, but our results do not depend on the specific choice we make (for instance, $B$ can be taken as $B=\sum_{i=1}^d(\var(A^i)+\langle M^i\rangle)$, with $\var(\cdot)$ denoting the total variation). Note that we do not necessarily assume that $S$ takes values in the positive orthant of $\R^d$. For every $t\in[0,T]$, let us denote by $c^+_t$ the Moore-Penrose pseudoinverse of the matrix $c_t$. The proof of Proposition 2.1 of \cite{DSpr} shows that the process $c^+=(c^+_t)_{0\leq t \leq T}$ is predictable and, hence, the process $a$ can be represented as
\be	\label{char-2}
a = c\,\lambda+\nu,
\ee
where $\lambda=(\lambda_t)_{0\leq t \leq T}$ is the $\R^d$-valued predictable process defined by $\lambda_t:=c^+_ta_t$, for all $t\in[0,T]$, and $\nu=(\nu_t)_{0\leq t \leq T}$ is an $\R^d$-valued predictable process with $\nu_t\in\text{Ker}(c_t):=\{x\in\R^d:c_tx=0\}$, for all $t\in[0,T]$.

We suppose that the financial market is frictionless, meaning that there are no trading restrictions, transaction costs, liquidity effects or other market imperfections. In order to mathematically describe the activity of trading, we need to introduce the notion of \emph{admissible strategy}. To this effect, let $L(S)$ be the set of all $\R^d$-valued $S$-integrable predictable processes, in the sense of \cite{JS}, and, for $H\in L(S)$, denote by $H\cdot S$ the stochastic integral process $\bigl(\int_0^t\!H_udS_u\bigr)_{0\leq t \leq T}$. Since $S$ is a continuous semimartingale, Proposition III.6.22 of \cite{JS} implies that $L(S)=L^2_{\text{loc}}(M)\cap L^0(A)$, where $L^2_{\text{loc}}(M)$ and $L^0(A)$ are the sets of all $\R^d$-valued predictable processes $H$ such that $\int_0^T\!H^{\top}_td\langle M,M\rangle_tH_t<\infty$ $P$-a.s. and $\int_0^T\!|H_t^{\top}dA_t|<\infty$ $P$-a.s., respectively. Hence, due to \eqref{char}, an $\R^d$-valued predictable process $H$ belongs to $L(S)$ if and only if
$$
\int_0^T\!v(H)_t\,dB_t<\infty \text{ $P$-a.s. }
\qquad\text{where}\quad
v(H)_t := \sum_{i,j=1}^d\!H^i_tc^{ij}_tH^j_t+\Bigl|\sum_{i=1}^dH^i_ta^i_t\Bigr|.
$$

\begin{Rem}	\label{integrands}
The set $L(S)$ represents the most general class of predictable integrands with respect to $S$. In particular, it contains non-locally bounded integrands, as in \cite{CMS}. Note that, for $H\in L\left(S\right)$, we have $H\cdot M\in\Mloc^c$ and the continuous semimartingale $H\cdot S$ admits the unique canonical decomposition $H\cdot S=H_0S_0+H\cdot A+H\cdot M$. We also want to emphasize that $H\cdot S$ has to be understood as the vector stochastic integral of $H$ with respect to $S$ and is in general different from the sum of the componentwise stochastic integrals $\sum_{i=1}^d\int\!H^i dS^i$; see for instance \cite{Ja} and \cite{SC}.
\end{Rem}

We are now in a position to formulate the following classical definition.

\begin{Def}	\label{strategy}
Let $a\in\R_+$. An element $H\in L(S)$ is said to be an \emph{$a$-admissible strategy} if $H_0=0$ and 
$(H\cdot S)_t\geq -a$ $P$-a.s. for all $t\in[0,T]$. An element $H\in L(S)$ is said to be an  \emph{admissible strategy} if it is an $a$-admissible strategy for some $a\in\R_+$.
\end{Def}

For $a\in\R_+$, we denote by $\A_a$ the set of all $a$-admissible strategies and by $\A$ the set of all admissible strategies, i.e., $\A=\bigcup_{a\in\R_+}\!\A_a$. As usual, $H^i_t$ represents the number of units of asset $i$ held in the portfolio at time $t$. The condition $H_0=0$ amounts to requiring that the initial position in the risky assets is zero and, hence, the initial endowment is entirely expressed in terms of the numéraire asset. For $H\in\A$, we define the \emph{gains from trading} process $G(H)=\bigl(G_t(H)\bigr)_{0\leq t \leq T}$ by $G_t(H):=(H\cdot S)_t$, for all $t\in[0,T]$.
According to Definition \ref{strategy}, the process $G\left(H\right)$ associated to an admissible strategy $H\in\A$ is uniformly bounded from below by some constant. This restriction is needed since the set $L\left(S\right)$ is too large for the purpose of modeling reasonable trading strategies and may also contain doubling strategies. This possibility is automatically ruled out if we impose a limit to the line of credit which can be granted to every market participant.
For $(x,H)\in\R_+\!\times\A$, we define the \emph{portfolio value} process $V(x,H)=\bigl(V_t(x,H)\bigr)_{0\leq t \leq T}$ by $V(x,H):=x+G(H)$. This corresponds to considering portfolios generated by \emph{self-financing} admissible strategies.

We now introduce five main notions of arbitrage. 

\begin{Def}	\label{arb} 
\mbox{}
\begin{itemize}
\item[(i)]
An element $H\in\A_0$ generates an \emph{increasing profit} if the process $G(H)$ is predictable\footnote{The reason for requiring $G(H)$ to be predictable will become clear in Theorem \ref{NIP}, which is formulated with respect to possibly discontinuous locally square-integrable semimartingales, in the sense of Definition II.2.27 of \cite{JS}. Of course, as soon as $S$ is continuous, the predictability requirement becomes unnecessary.} and if $P\bigl(G_s(H)\leq G_t(H),\text{ for all }0\leq s\leq t\leq T\bigr)=1$ and $P\bigl(G_T(H)>0\bigr)>0$. If there exists no such $H\in\A_0$ we say that the \emph{No Increasing Profit (NIP)} condition holds;
\item[(ii)]
an element $H\in\A_0$ generates a \emph{strong arbitrage opportunity} if $P\bigl(G_T(H)>0\bigr)$ $>0$. If there exists no such $H\in\A_0$, i.e., if $\bigl\{G_T(H):H\in\A_0\bigr\}\cap L^0_+=\{0\}$, we say that the \emph{No Strong Arbitrage (NSA)} condition holds;
\item[(iii)]
a non-negative random variable $\xi$ generates an \emph{arbitrage of the first kind} if $P(\xi>0)>0$ and for every $v\in(0,\infty)$ there exists an element $H^v\in\A_v$ such that $V_T(v,H^v)\geq\xi$ $P$-a.s. If there exists no such random variable $\xi$ we say that the \emph{No Arbitrage of the First Kind (NA1)} condition holds;
\item[(iv)]
an element $H\in\A$ generates an \emph{arbitrage opportunity} if $G_T(H)\geq0$ $P$-a.s. and  $P\bigl(G_T(H)>0\bigr)>0$. If there exists no such $H\in\A$, i.e., if $\bigl\{G_T(H):H\in\A\bigr\}\cap L^0_+=\{0\}$, we say that the \emph{No Arbitrage (NA)} condition holds;
\item[(v)]
a sequence $\{H^n\}_{n\in\N}\subset\A$ generates a \emph{Free Lunch with Vanishing Risk} if there exist an $\varepsilon>0$ and an increasing sequence $\{\delta_n\}_{n\in\N}$ with $0\leq\delta_n\nearrow 1$ such that $P\bigl(G_T(H^n)>-1+\delta_n\bigr)=1$ and $P\bigl(G_T(H^n)>\varepsilon\bigr)\geq\varepsilon$, for all $n\in\N$. If there exists no such sequence we say that the \emph{No Free Lunch with Vanishing Risk (NFLVR)} condition holds.
\end{itemize}
\end{Def}

The NIP condition is similar to the No Unbounded Increasing Profit condition introduced in \cite{KK} and represents the strongest notion of arbitrage among those listed above. The ``unboundedness'' in the original definition of \cite{KK} can be explained as follows: if $H\in\A_0$ yields an increasing profit in the sense of Definition \ref{arb}-\emph{(i)}, we have $H^n:=nH\in\A_0$ and $G(H^n)\geq G(H)$, for every $n\in\N$. This means that the increasing profit generated by $H$ can be scaled to arbitrarily large levels of wealth. The NSA condition corresponds to the notion of absence of arbitrage opportunities adopted in Section 3 of \cite{LW} as well as to the NA$^+$ condition studied in \cite{Str}. The above formulation of the notion of arbitrage of the first kind has been introduced by \cite{Ka1}. The NA and NFLVR conditions are classical and, in particular, go back to the seminal papers \cite{HK}, \cite{HP} and \cite{DS94}. Note that the NA condition can be equivalently formulated as $\mathcal{C}\cap L^{\infty}_+=\{0\}$, where $\mathcal{C}:=\left(\left\{G_T(H):H\in\A\right\}-L^0_+\right)\cap L^{\infty}$. The above definition of NFLVR is taken from \cite{KK} and can be shown to be equivalent to $\overline{\mathcal{C}}\cap L^{\infty}_+=\{0\}$, with the bar denoting the closure in the norm topology of $L^{\infty}$, as in \cite{DS94}. In the following sections, we shall also examine several other no-arbitrage conditions equivalent to the ones introduced in Definition \ref{arb}.

\section{No Increasing Profit}	\label{S3}

An increasing profit represents an investment opportunity which does not require any initial investment nor any line of credit and, moreover, generates an increasing wealth process, yielding a non-zero final wealth with strictly positive probability. As such, the notion of increasing profit represents the most egregious form of arbitrage and, therefore, should be banned from any reasonable financial model. The following theorem characterises the NIP condition. At no extra cost, we state and prove the result for general locally square-integrable semimartingales, in the sense of Definition II.2.27 of \cite{JS}.

\begin{Thm}	\label{NIP}
The following are equivalent, using the notation introduced in \eqref{char}-\eqref{char-2}:
\begin{itemize}
\item[(i)]
the NIP condition holds;
\item[(ii)]
for every $H\in L(S)$, if $H^{\top}_tc_t=0$ $P\otimes B$-a.e. then $H^{\top}_ta_t=0$ $P\otimes B$-a.e.;
\item[(iii)]
$\nu_t=0$ $P\otimes B$-a.e.
\end{itemize}
\end{Thm}
\begin{proof}
$(i)\Rightarrow (ii)$: Let us define the product space $\overline{\Omega}:=\Omega\times[0,T]$. Suppose that NIP holds and let $H=(H_t)_{0\leq t \leq T}$ be a process in $L(S)$ such that $H_t^{\top}c_t=0$ $P\otimes B$-a.e. (so that $H\cdot M=0$) but $P\otimes B\bigl((\omega,t)\in\overline{\Omega}:H_t^{\top}(\omega)a_t(\omega)\neq 0\bigr)>0$. By the Hahn-Jordan decomposition (see \cite{DS95b}, Theorem 2.1), we can write $H\cdot A=\int(\ind_{D^+}-\ind_{D^-})dV$, where $D^+$ and $D^-$ are two disjoint predictable subsets of $\overline{\Omega}$ such that $D^+\cup D^-=\overline{\Omega}$ and $V:=\var(H\cdot A)$. Let $\psi:=\ind_{D^+}-\ind_{D^-}$ and define the $\R^d$-valued predictable process $\tilde{H}:=\psi H\ind_{\lsir0,T\rsi}$. Due to the associativity of the stochastic integral, it is clear that $\tilde{H}\in L(S)$ and $\tilde{H}\cdot M=0$. Thus, using again the associativity of the stochastic integral,
$$
\tilde{H}\cdot S = \tilde{H}\cdot A = (\psi H)\cdot A
= \psi\cdot (H\cdot A) = \psi^2 \cdot V = V.
$$
The process $V$ is non-negative, increasing and predictable and satisfies $P\left(V_T>0\right)>0$, since $H\cdot A$ is supposed to be not identically zero. Clearly, this amounts to saying that $\tilde{H}$ generates an increasing profit, thus contradicting the assumption that NIP holds. Hence, it must be $H_t^{\top}a_t=0$ $P\otimes B$-a.e. \\
$(ii)\Rightarrow(iii)$: let $H=(H_t)_{0\leq t \leq T}$ be an $\R^d$-valued predictable process such that $\|H_t(\omega)\|\in\{0,1\}$ for all $(\omega,t)\in\overline{\Omega}$. Since $H_t^{\top}c_t=0$ $P\otimes B$-a.e. implies that $H_t^{\top}a_t=0$ $P\otimes B$-a.e., condition \emph{(iii)} follows directly from the absolute continuity result of Theorem 2.3 of \cite{DS95b}.	\\
$(iii)\Rightarrow(i)$: suppose that $\nu_t=0$ $P\otimes B$-a.e. and let $H\in\A_0$ generate an increasing profit. The process $G(H)=H\cdot S$ is predictable and increasing, hence of finite variation. In particular, $H\cdot S$ is a special semimartingale and, hence, due to Proposition 2 of \cite{Ja}, we can write $H\cdot S=H\cdot A+H\cdot M$. This implies that $H\cdot M=H\cdot S-H\cdot A$ is also predictable and of finite variation. Theorem III.15 of \cite{Pr} then implies that $H\cdot M=0$, being a predictable local martingale of finite variation. Hence:
$$
G_t(H) 
= (H\cdot A)_t
= \int_0^t\!H^{\top}_ua_u\,dB_u
= \int_0^t\!H^{\top}_uc_u\lambda_u\,dB_u
= \int_0^t\!d\langle H\cdot M,M\rangle_{\!u}\lambda_u = 0
\qquad\; \text{$P$-a.s.}
$$
for all $t\in[0,T]$. In particular, $P\bigl(G_T(H)>0\bigr)=0$, thus contradicting the hypothesis that $H$ generates an increasing profit. 
\end{proof}

Clearly, $\nu_t=0$ $P\otimes B$-a.e. means that $dA\ll d\langle M,M\rangle$. The latter condition is known in the literature as the \emph{weak structure condition} and the process $\lambda$ is usually referred to as the \emph{instantaneous market price of risk} (see e.g. \cite{HS}, Section 3). We want to point out that results similar to Theorem \ref{NIP} have already appeared in the literature, albeit under stronger assumptions. In particular, Theorem 3.5 of \cite{DS95b} (see also \cite{PS}, Theorem 1, and Appendix B of \cite{KarShr}) shows that $dA\ll d\langle M,M\rangle$ holds under the classical NA condition, which is strictly stronger than NIP (see Section \ref{S6}). Somewhat more generally, \cite{KS} (see also \cite{Hulley}, Theorem 1.13) prove that $dA\ll d\langle M,M\rangle$ holds under the NSA condition, which is also strictly stronger than the NIP condition (see Section \ref{S4}). Theorem \ref{NIP} shows that the weak structure condition $dA\ll d\langle M,M\rangle$ is exactly equivalent to the NIP condition, which represents an indispensable requirement for any reasonable financial market.

\begin{Rems}	\label{rem-min}
{\textbf{1)} }
As can be seen by inspecting the proof of Theorem \ref{NIP}, the NIP condition is also equivalent to the absence of elements $H\in\A_0$ such that the gains from trading process $G(H)$ is predictable and of finite variation (not necessarily increasing) and satisfies $P(G_T(H)>0)>0$.

{\textbf{2)} }
In general, as long as the NIP condition holds, there may exist many $\R^d$-valued predictable processes $\gamma=(\gamma_t)_{0\leq t \leq T}$ such that $dA_t=d\langle M,M\rangle_{\!t}\gamma_t$. However, for any such process $\gamma$, elementary linear algebra gives $\Pi_{c_t}(\gamma_t)=\lambda_t$, where we denote by $\Pi_{c_t}(\cdot)$ the orthogonal projection onto the range of the matrix $c_t$, for $t\in[0,T]$. In turn, this implies that $\int_0^T\!\!\gamma_t^{\top}\!c_t\gamma_tdB_t\geq\int_0^T\!\!\lambda_t^{\top}\!c_t\lambda_tdB_t$, thus showing the \emph{minimality} property of the process $\lambda$ introduced in \eqref{char-2}.
\end{Rems}

\section{No Strong Arbitrage}	\label{S4}

A strong arbitrage opportunity consists of an investment opportunity which does not require any initial capital nor any line of credit and leads to a non-zero final wealth with strictly positive probability. Of course, this sort of strategy should be banned from any reasonable financial market, since every agent would otherwise benefit in an unlimited way from a strong arbitrage opportunity (compare also with the discussion in Section \ref{S7}). According to Definition \ref{arb}, it is evident that an increasing profit generates a strong arbitrage opportunity. Two examples of models which satisfy NIP but allow for strong arbitrage opportunities will be presented in Section \ref{examples}, thus showing that NSA is strictly stronger than NIP.

Let us now introduce another notion of arbitrage, which has been first formulated in \cite{DS95b} and turns out to be equivalent to the notion of strong arbitrage opportunity.

\begin{Def}	\label{IA}
An element $H\in\A_0$ generates an \emph{immediate arbitrage opportunity} if there exists a stopping time $\tau$ such that $P(\tau<T)>0$ and if $H=H\ind_{\lsir\tau,T\rsi}$ and $G_t(H)>0$ $P$-a.s. for all $t\in(\tau,T]$. If there exists no such $H\in\A_0$ we say that the \emph{No Immediate Arbitrage (NIA)} condition holds.
\end{Def}

We then have the following simple lemma (compare also with \cite{DS95b}, Lemma 3.1).

\begin{Lem}	\label{NIAO-NSA}
The NSA condition and the NIA condition are equivalent.
\end{Lem}
\begin{proof}
Suppose that $H\in\A_0$ generates a strong arbitrage opportunity and define the stopping time $\tau:=\inf\{t\in[0,T]:G_t(H)>0\}\wedge T$. Since $P\left(G_T(H)>0\right)>0$, we have $P(\tau<T)>0$. For a sequence $\{\theta_n\}_{n\in\N}$ dense in $(0,1)$, let us define the process $\tilde{H}:=\sum_{n=1}^{\infty}2^{-n}H\ind_{\lsir\tau,\,(\tau+\theta_n)\wedge T\rsi}$. Clearly, we have $\tilde{H}\in\A_0$. Furthermore, on the event $\{\tau<T\}$ it holds that, for every $\varepsilon>0$,
$$
G_{\tau+\varepsilon}(\tilde{H})
= (\tilde{H}\cdot S)_{\tau+\varepsilon}
= \sum_{n=1}^{\infty}\frac{1}{2^n}\left((H\cdot S)_{\left(\tau+(\varepsilon\wedge\theta_n)\right)\wedge T}-(H\cdot S)_{\tau}\right)
= \sum_{n=1}^{\infty}\frac{1}{2^n}(H\cdot S)_{\left(\tau+(\varepsilon\wedge\theta_n)\right)\wedge T} > 0
\; \text{$P$-a.s.}
$$
thus showing that $\tilde{H}$ generates an immediate arbitrage opportunity at the stopping time $\tau$. Conversely, Definitions \ref{arb}-\emph{(ii)} and \ref{IA} directly imply that an immediate arbitrage opportunity is also a strong arbitrage opportunity.
\end{proof}

Recall that, due to Theorem \ref{NIP}, the NIP condition is equivalent to $a=c\,\lambda$ $P\otimes B$-a.e., where the processes $a$, $c$, $\lambda$ and $B$ are as in \eqref{char}-\eqref{char-2}. Since NSA (or, equivalently, NIA) is stronger than NIP, it is natural to expect that NSA will imply some additional properties of the process $\lambda$. This is confirmed by the next theorem. As a preliminary, let us define the \emph{mean-variance trade-off process} $\hatK=(\hatK_t)_{0\leq t \leq T}$ as
\be	\label{MVT}
\hatK_t := \int_0^t\!\lambda^{\top}_ud\langle M,M\rangle_{\!u}\lambda_u
= \int_0^t\!\lambda^{\top}_uc_u\lambda_u\,dB_u
= \int_0^t\!a^{\top}_uc^+_ua_u\,dB_u,
\qquad \text{for }t\in[0,T].
\ee
Let also $\hatK_s^t:=\hatK_t-\hatK_s$, for $s,t\in\left[0,T\right]$ with $s\leq t$. Following \cite{LS} and \cite{Str}, we also define the stopping time
$$
\alpha:=\inf\bigl\{t\in[0,T]:\hatK_t^{t+h}=\infty,\forall h\in\left(0,T-t\right]\bigr\},
$$
with the usual convention $\inf\emptyset=\infty$.
The next theorem is essentially due to \cite{Str}, but we opt for a slightly different proof.

\begin{Thm}	\label{NSA}
The NSA condition holds if and only if $\nu_t=0$ $P\otimes B$-a.e. and $\alpha=\infty$ $P$-a.s.
\end{Thm}
\begin{proof}
Suppose first that NSA holds. Since NSA implies NIP, Theorem \ref{NIP} gives that $\nu_t=0$ $P\otimes B$-a.e. The fact that $\alpha=\infty$ $P$-a.s. then follows from Theorem 3.6 of \cite{DS95b} together with Lemma \ref{NIAO-NSA} (compare also with \cite{KS}, Sections 3-4). \\
Conversely, suppose that $\nu_t=0$ $P\otimes B$-a.e. and $\alpha=\infty$ $P$-a.s. and let $H\in\A_0$ generate a strong arbitrage opportunity. Due to Lemma \ref{NIAO-NSA}, we can equivalently suppose that $H$ generates an immediate arbitrage opportunity with respect to a stopping time $\tau$ with $P(\tau<T)>0$. Since $P(\alpha=\infty)=1$, we have $P\bigl(\hatK_{\tau}^{\tau+h}=\infty,\forall h\in(0,T-\tau]\bigr)=0$. For each $n\in\N$, define the stopping time $\rho_n:=\inf\bigl\{t>\tau:\hatK_{\tau}^t\geq n\bigr\}\wedge T$. Since $\hatK$ is continuous and does not jump to infinity, it is clear that $\rho_n>\tau$ $P$-a.s. on the set $\{\tau<T\}$, for all $n\in\N$. Let us then define the predictable process $\lambda^{\tau,n}:=\lambda\ind_{\lsir\tau,\rho_n\rsi}$, for every $n\in\N$. Then, on the set $\{\tau<T\}$,
$$
\int_0^T\!(\lambda^{\tau,n}_t)^{\top}d\langle M,M\rangle_{\!t}\lambda^{\tau,n}_t
= \int_0^T\!\!\ind_{\lsir\tau,\rho_n\rsi}\lambda^{\top}_td\langle M,M\rangle_{\!t}\lambda_t
= \hatK_{\tau}^{\rho_n}\leq n
\qquad\text{$P$-a.s.}
$$
For every $n\in\N$, we can define the stochastic exponential $\hatZ^{\tau,n}:=\mathcal{E}(-\lambda^{\tau,n}\cdot M)$ as a strictly positive process in $\M^c$, due to Novikov's condition. It is obvious that $\hatZ^{\tau,n}=1$ on $\lsi 0,\tau\rsi$ and $\hatZ^{\tau,n}=\hatZ^{\tau,n}_{\rho_n}$ on $\lsi\rho_n,T\rsi$. We now apply the integration by parts formula to $\hatZ^{\tau,n}(H\cdot S)^{\rho_n}$, where $(H\cdot S)^{\rho_n}$ denotes the process $H\cdot S$ stopped at $\rho_n$, and use the fact that $\hatZ^{\tau,n}_td(H\cdot A)^{\rho_n}_t=\hatZ^{\tau,n}_tH_t^{\top}dA^{\rho_n}_t$, since $\hatZ^{\tau,n}\in L(H\cdot A)$ (being $\hatZ^{\tau,n}$ adapted and continuous, hence predictable and locally bounded), and the fact that $dA=d\langle M,M\rangle\lambda$ and $H=H\ind_{\lsir\tau,T\rsi}$:
$$	\ba
d\bigl(\hatZ^{\tau,n}_t(H\cdot S)^{\rho_n}_t\bigr)
&= \hatZ^{\tau,n}_td(H\cdot S)^{\rho_n}_t+(H\cdot S)^{\rho_n}_td\hatZ^{\tau,n}_t
+d\bigl\langle\hatZ^{\tau,n},(H\cdot S)^{\rho_n}\bigr\rangle_{\!t}	\\
&= \hatZ^{\tau,n}_td(H\cdot M)^{\rho_n}_t+\hatZ^{\tau,n}_td\left(H\cdot A\right)^{\rho_n}_t
+(H\cdot S)^{\rho_n}_td\hatZ^{\tau,n}_t
-\hatZ^{\tau,n}_tH_t^{\top}d\langle M,M\rangle_{\!t}\lambda^{\tau,n}_t	\\
&= \hatZ^{\tau,n}_td(H\cdot M)^{\rho_n}_t+(H\cdot S)^{\rho_n}_td\hatZ^{\tau,n}_t
+\hatZ^{\tau,n}_tH_t^{\top}\bigl(dA^{\rho_n}_t-d\langle M,M\rangle_{\!t}\lambda^{\tau,n}_t\bigr)	\\
&= \hatZ^{\tau,n}_td(H\cdot M)^{\rho_n}_t+(H\cdot S)^{\rho_n}_td\hatZ^{\tau,n}_t.
\ea	$$
This shows that $\hatZ^{\tau,n}(H\cdot S)^{\rho_n}$ is a non-negative local martingale and, by Fatou's lemma, also a supermartingale, for every $n\in\N$. Since $\hatZ^{\tau,n}_0(H\cdot S)^{\rho_n}_0=0$, the supermartingale property implies that $\hatZ^{\tau,n}_t(H\cdot S)^{\rho_n}_t=0$ for all $t\in[0,T]$ $P$-a.s., meaning that $H\cdot S=0$ $P$-a.s. on $\bigcup_{n\in\N}\lsi0,\rho_n\rsi$. Since $\rho_n>\tau$ $P$-a.s. on $\{\tau<T\}$ and $P(\tau<T)>0$, this contradicts the fact that $(H\cdot S)_t>0$ $P$-a.s. for all $t\in(\tau,T]$, thus showing that there cannot exist an immediate arbitrage opportunity. Equivalently, due to Lemma \ref{NIAO-NSA}, the NSA condition holds.
\end{proof}

Theorem \ref{NSA} shows that NSA holds as long as the mean-variance trade-off process $\hatK$ does not jump to infinity (however, $\hatK_T$ is not guaranteed to be finite). In particular, it is important to remark that we can check whether a financial market allows for strong arbitrage opportunities by looking only at the semimartingale characteristics of the discounted price process $S$. 

We now introduce the important concept of \emph{(weak) martingale deflator}, which represents a weaker counterpart to the density process of an equivalent local martingale measure (see Remark \ref{ELMM-defl}) and corresponds to the notion of \emph{martingale density} introduced in \cite{S1}-\cite{S2}.

\begin{Def}	\label{weak-defl}
Let $Z=\left(Z_t\right)_{0\leq t \leq T}$ be a non-negative local martingale with $Z_0=1$. We say that $Z$ is a \emph{weak martingale deflator} if the product $ZS^i$ is a local martingale, for all $i=1,\ldots,d$. If $Z$ satisfies in addition $Z_T>0$ $P$-a.s. we say that $Z$ is a \emph{martingale deflator}.	\\
A (weak) martingale deflator $Z$ is said to be \emph{tradable} if there exists a sequence $\{\theta^n\}_{n\in\N}\subseteq\A_1$ and a sequence $\{\tau_n\}_{n\in\N}$ of stopping times increasing $P$-a.s. to $\tau:=\inf\bigl\{t\in[0,T]:Z_t=0\text{ or }Z_{t-}=0\bigr\}$ such that $1/Z^{\tau_n}=V(1,\theta^n)$ $P$-a.s., for every $n\in\N$.
\end{Def}

\begin{Rem}	\label{Remdefl}
Fatou's lemma implies that any weak martingale deflator $Z$ is a supermartingale (and also a true martingale if and only if $E[Z_T]=1$). Furthermore, if $Z$ is a martingale deflator, so that $Z_T>0$ $P$-a.s., the minimum principle for non-negative supermartingales (see e.g. \cite{RY}, Proposition II.3.4) implies that $\tau=\infty$ $P$-a.s. It can be verified that a martingale deflator is tradable if and only if there exists a strategy $\theta\in\A_1$ such that $1/Z=V(1,\theta)$ (indeed, it suffices to define $\theta:=\sum_{n=1}^{\infty}\!\theta^n\ind_{\lsir\tau_{n-1},\tau_n\rsi}$, with $\tau_0:=0$). This also explains the meaning of the terminology \emph{tradable}\footnote{To the best of our knowledge, for a weak martingale deflator $Z$, the definition of tradability as in Definition \ref{weak-defl} seems to be new, but is related to condition {\bf H} from \cite{KS}.}.
\end{Rem}

We denote by $\Dweak$ and $\D$ the families of all weak martingale deflators and of all martingale deflators, respectively. The next lemma shows the fundamental property of (weak) martingale deflators. At little extra cost, we state and prove the result for the case of general (possibly discontinuous and non-locally bounded) semimartingales (we refer to Section III.6e of \cite{JS} and to \cite{Kall} for the definition and the main properties of $\sigma$-martingales). The result is more or less well-known but, for the convenience of the reader, we give a detailed proof in the Appendix.

\begin{Lem}	\label{sigma}
Let $Z\in\mDweak$. Then, for any $H\in L(S)$, the product $Z(H\cdot S)$ is a $\sigma$-martingale.
\linebreak If in addition $H\in\A$, then $Z(H\cdot S)\in\mMloc$.
\end{Lem}

If $Z\in\Dweak$ and $H\in\A_1$, Lemma \ref{sigma} implies that $Z(1+H\cdot S)$ is a non-negative local martingale and, hence, a supermartingale. This means that $Z$ is a \emph{$P$-supermartingale density}, according to the terminology adopted in \cite{Be}. If we also have $Z_T>0$ $P$-a.s., i.e., $Z\in\D$, then $Z$ is an \emph{equivalent supermartingale deflator} in the sense of Definition 4.9 of \cite{KK}. The importance of supermartingale densities/deflators has been first recognized by \cite{KrSch} in the context of utility maximisation.

We now show that the NSA condition ensures the existence of a tradable weak martingale deflator. This can already be guessed by carefully examining the proof of Theorem \ref{NSA}, but, since the result is of interest, we prefer to give full details.

\begin{Prop}	\label{weakdfl}
Let $\tau:=\inf\bigl\{t\in[0,T]:\hatK_t=\infty\bigr\}$. If the NSA condition holds then the process $\hatZ:=\mathcal{E}(-\lambda\cdot M)\ind_{\lsi0,\tau\rsir}$ is a tradable weak martingale deflator. Furthermore, $\hatZ N\in\mMloc$ for any $N=(N_t)_{0\leq t \leq T}\in\mMloc$ orthogonal to $M$ (in the sense of \cite{JS}, Definition I.4.11).
\end{Prop}
\begin{proof}
Note first that, due to Theorem \ref{NSA}, we have $\tau>0$ $P$-a.s. Furthermore, the sequence $\{\tau_n\}_{n\in\N}$, defined as $\tau_n:=\inf\{t\in[0,T]:\hatK_t\geq n\}$, $n\in\N$, is an announcing sequence for $\tau$, in the sense of I.2.16 of \cite{JS}, and we have $\lsi0,\tau\rsir=\bigcup_{n\in\N}\lsi0,\tau_n\rsi$. Since $\hatK_{T\wedge\tau_n}\leq n$ $P$-a.s. for every $n\in\N$, the process $\hatZ:=\mathcal{E}(-\lambda\cdot M)$ is well-defined as a continuous local martingale on $\lsi0,\tau\rsir$, in the sense of Section 5.1 of \cite{Jacod}. On $\{\tau\leq T\}$ , we have $\hatK_{\tau}=\infty$ and $\hatZ_{\tau-}=0$ $P$-a.s. 
By letting $\hatZ=\hatZ_{\tau-}=0$ on $\lsi\tau,T\rsi$, $\hatZ$ can be extended to a continuous local martingale on the whole interval $[0,T]$.  Furthermore, the integration by parts formula gives that, for every $i=1,\ldots,d$,
$$	\ba
d(\hatZ S^i)_t
&= \hatZ_t\,dS^i_t+S^i_t\,d\hatZ_t+d\langle\hatZ,S^i\rangle_{\!t}
= \hatZ_t\,dM^i_t+\hatZ_t\,d\langle M^i,M\rangle_t\lambda_t+S^i_t\,d\hatZ_t-\hatZ_t\lambda_t^{\top}d\langle M,M^i\rangle_{\!t}	\\
&= \hatZ_t\,dM^i_t+S^i_t\,d\hatZ_t.
\ea	$$
Since $S^i$ and $\hatZ$ are continuous, this implies that $\hatZ S^i\in\Mloc^c$, for every $i=1,\ldots,d$. We have thus shown that $\hatZ=\mathcal{E}(-\lambda\cdot M)\ind_{\lsi0,\tau\rsir}\in\Dweak$. 
To prove the tradability of $\hatZ$, note that the process $1/\hatZ$ is well defined on $\lsi0,\tau\rsir=\bigcup_{n\in\N}\lsi0,\tau_n\rsi$. It\^o's formula gives then the following, for every $n\in\N$:
\be	\label{tradability}
d\frac{1}{\hatZ^{\tau_n}_t}
= -\frac{1}{\bigl(\hatZ^{\tau_n}_t\bigr)^2}\,d\hatZ^{\tau_n}_t + \frac{1}{\bigl(\hatZ^{\tau_n}_t\bigr)^3}\,
d\langle\hatZ\rangle^{\tau_n}_{\!t}
= \frac{1}{\hatZ^{\tau_n}_t}\lambda_t\,dM^{\tau_n}_t+\frac{1}{\hatZ^{\tau_n}_t}\lambda_t^{\top}d\langle M,M\rangle_{\!t}^{\tau_n}\lambda_t
= \theta_t^ndS_t,
\ee
where $\theta^n\!:=\!\ind_{\lsir0,\tau_n\rsi}\lambda\hatZ^{-1}\!\in\!\A_1$, for all $n\!\in\!\N$. 
Finally, for any $N\!=\!(N_t)_{0\leq t \leq T}\!\in\!\Mloc$ orthogonal to $M$:
$$
\hatZ N = N_0+\hatZ\cdot N + N_-\cdot\hatZ + \langle\hatZ,N\rangle
= N_0+\hatZ\cdot N + N_-\cdot\hatZ - \hatZ\lambda\cdot\langle M,N\rangle
= N_0+\hatZ\cdot N + N_-\cdot\hatZ,
$$
where we have used the continuity of $\hatZ$ and the orthogonality of $M$ and $N$. Since $\hatZ$ and $N_-$ are predictable and locally bounded, being adapted and left-continuous, and since $N,\hatZ\in\Mloc$, Theorem IV.29 of \cite{Pr} implies that $\hatZ N\in\Mloc$.
\end{proof}

\begin{Rem}[\emph{On the minimal martingale measure}]	\label{weak-MMM}
The process $\hatZ$ is the candidate density process of the \emph{minimal martingale measure}, originally introduced in \cite{FS} and defined as a probability measure $\hatQ\sim P$ on $(\Omega,\F)$ with $\hatQ=P$ on $\F_0$ such that $S$ is a local $\hatQ$-martingale and every local $P$-martingale orthogonal to the martingale part $M$ in the canonical decomposition of $S$ (with respect to $P$) remains a local $\hatQ$-martingale. However, even if NSA holds, the process $\hatZ$ can fail to be a well-defined density process for two reasons. First, if $P(\hatZ_T>0)<1$, the measure $\hatQ$ defined by $d\hatQ:=\hatZ_T\,dP$ fails to be equivalent to $P$, being only absolutely continuous. Second, $\hatZ$ may fail to be a true martingale, being instead a \emph{strict} local martingale in the sense of \cite{ELY}, i.e., a local martingale which is not a true martingale, so that $E[\hatZ_T]<E[\hatZ_0]=1$. In the latter case, $\hatQ$ fails to be a probability measure, since $\hatQ(\Omega)=E[\hatZ_T]<1$.
\end{Rem}

\begin{Rem}
A strong arbitrage opportunity corresponds to the notion of arbitrage adopted in the context of the Benchmark Approach, see e.g. Section 7 of \cite{Pl3} and Section 10.3 of \cite{PH}. However, we want to make the reader aware of the fact that typical applications of the Benchmark Approach require assumptions stronger than NSA, namely the existence of the \emph{Growth Optimal Portfolio (GOP)}. Theorem 4.12 of \cite{KK} shows that the existence of a (non-exploding) GOP is equivalent to the \emph{No Unbounded Profit with Bounded Risk (NUPBR)} condition, which is strictly stronger than NSA (see e.g. Example \ref{example-NA1}). Hence, in the context of the Benchmark Approach, not only strong arbitrage opportunities but also slightly weaker forms of arbitrage must be ruled out from the market (see also Section \ref{S7} for a related discussion). The NSA condition has been also adopted in \cite{CL} as a necessary (but not sufficient) requirement in order to construct the GOP. 
\end{Rem}

\section{No Arbitrage of the First Kind}	\label{S5}

An arbitrage of the first kind amounts to a non-negative and non-zero payoff which can be super-replicated via a non-negative portfolio by every market participant, regardless of his/her initial wealth. It is evident that a strong arbitrage opportunity yields an arbitrage of the first kind. Indeed, let $H\in\A_0$ generate a strong arbitrage opportunity and define $\xi:=G_T(H)$. By Definition \ref{arb}-\emph{(ii)}, it holds that $P(\xi\geq 0)=1$ and $P(\xi>0)>0$. Moreover, for any $v\in(0,\infty)$, we also have $V_T(v,H)=v+G_T(H)>\xi$, thus showing that $\xi$ generates an arbitrage of the first kind. A model satisfying NSA but allowing for arbitrages of the first kind will be presented in Example \ref{example-NA1}, thus showing that NA1 is strictly stronger than NSA.

We now introduce two alternative notions of arbitrage which will be shown to be equivalent to an arbitrage of the first kind (Lemma \ref{A1-UPBR-CT}).

\begin{Def}	\label{UPBR} \mbox{}
\begin{itemize}
\item[(i)]
A sequence $\{H^n\}_{n\in\N}\subset\A_1$ generates an \emph{unbounded profit with bounded risk} if the collection $\{G_T(H^n)\}_{n\in\N}$ is unbounded in probability, i.e., if  $\lim_{m\rightarrow\infty}\sup_{n\in\N}P(G_T(H^n)>m)>0$. If there exists no such sequence we say that the \emph{No Unbounded Profit with Bounded Risk (NUPBR)} condition holds;
\item[(ii)]
let $\{x_n\}_{n\in\N}\subset\R_+$ be a sequence such that $x_n\searrow0$ as $n\rightarrow\infty$. A sequence  $\{H^n\}_{n\in\N}\subset\A$ with $H^n\in\A_{x_n}$, for all $n\in\N$, generates a \emph{cheap thrill} if $V_T(x_n,H^n)\rightarrow\infty$ $P$-a.s. as $n\rightarrow\infty$ on some event with strictly positive probability. If there exists no such sequence we say that the \emph{No Cheap Thrill (NCT)} condition holds.
\end{itemize}
\end{Def}

The NUPBR condition has been first introduced under that name in \cite{KK} and corresponds to condition BK in \cite{Kab} (the same condition also plays a key role in the seminal paper \cite{DS94}). Note that there is no loss of generality in considering $1$-admissible strategies in Definition \ref{UPBR}-\emph{(i)}. Indeed, we have $\{G_T(H):H\in\A_a\}=a\{G_T(H):H\in\A_1\}$, for any $a>0$, and, hence, the set of all final wealths generated by $a$-admissible strategies is bounded in probability if and only if the set of all final wealths generated by $1$-admissible strategies is bounded in probability. 
The notion of cheap thrill has been introduced by \cite{LW} in the context of a complete It\^o process model and can easily be shown to be equivalent to the notion of \emph{asymptotic arbitrage of the first kind} (with respect to the fixed probability measure $P$) of \cite{KabKra}, hence the name ``arbitrage of the first kind'' in Definition \ref{arb}-\emph{(iii)}\footnote{We want to point out that a cheap thrill is also equivalent to an \emph{approximate arbitrage} in the sense of \cite{CL}, as the reader can easily verify. However, we shall use the term ``approximate arbitrage'' with a different meaning in Section \ref{S6}.}.

The next lemma proves the equivalence between the notions introduced in Definition \ref{UPBR} and the notion of arbitrage of the first kind. The proof relies on techniques similar to those used in Section 3 of \cite{DS94} or in Proposition 1 of \cite{Ka1} and does not rely on the continuity of $S$.

\begin{Lem}	\label{A1-UPBR-CT}
The NA1, NUPBR and NCT conditions are all equivalent.
\end{Lem}
\begin{proof}
Let the random variable $\xi$ generate an arbitrage of the first kind. By definition, for every $n\in\N$, there exists a strategy $H^n\in\A_{1/n}$ such that $V_T(1/n,H^n)\geq\xi$ $P$-a.s. For every $n\in\N$, define $\tildeH^n:=nH^n$, so that $\{\tildeH^n\}_{n\in\N}\subset\A_1$ and $G_T(\tildeH^n)=nG_T(H^n)\geq n\xi-1$ $P$-a.s. Since $P(\xi>0)>0$, this implies that the collection $\{G_T(\tildeH^n):n\in\N\}$ is unbounded in probability.	\\
Let $\{H^n\}_{n\in\N}\subset\A_1$ generate an unbounded profit with bounded risk, so that $P\bigl(G_T(H^n)\geq n\bigr)>\beta$ for all $n\in\N$ and for some $\beta>0$. Let $\tildeH^n:=n^{-1}H^n$, for every $n\in\N$, so that $\tildeH^n\in\A_{1/n}$ and $P\bigl(G_T(\tildeH^n)\geq 1\bigr)>\beta$. Let $f_n:=n^{-1}+G_T(\tildeH^n)\geq 0$ $P$-a.s., for all $n\in\N$. Due to Lemma A1.1 of \cite{DS94}, there exists a sequence $\{g_n\}_{n\in\N}$, with $g_n\in\conv\{f_n,f_{n+1},\ldots\}$, such that $\{g_n\}_{n\in\N}$ converges $P$-a.s. to a non-negative random variable $g$ as $n\rightarrow\infty$. For all $n\in\N$, let $K^n$ be the convex combination of strategies $\{\tildeH^m\}_{m\geq n}$ corresponding to $g_n$. It is easy to check that $K^n\in\A_{1/n}$, for every $n\in\N$. Furthermore, we have $G_T(K^n)=g_n+O\left(n^{-1}\right)$, so that $G_T\left(K^n\right)\rightarrow g$ $P$-a.s. as $n\rightarrow\infty$. The last assertion of Lemma A1.1 of \cite{DS94} implies that $P\left(g>0\right)>0$. By letting $x_n:=\log(n)/n$ and $\tildeK^n:=\log(n)K^n$, for every $n\in\N$, so that $\tildeK^n\in\A_{x_n}$, we then obtain a sequence $\{\tildeK^n\}_{n\in\N}$ which generates a cheap thrill.	\\
Finally, let the sequence $\{H^n\}_{n\in\N}$ generate a cheap thrill, with respect to $\{x_n\}_{n\in\N}$. By definition, this implies that, for each $n\in\N$, the set $C_n:=\bigl\{V_T(x_m,H^m):m\in\N,m\geq n\bigr\}$ is hereditarily unbounded in probability on $\Omega_u:=\bigl\{\omega\in\Omega:\lim_{n\rightarrow\infty}V_T(x_n,H^n)(\omega)=\infty\bigr\}$, in the sense of \cite{BS}. Then $\tildeC_n:=\conv C_n$ is hereditarily unbounded in probability on $\Omega_u$ as well, for all $n\in\N$. Similarly as in the proof of Proposition 1 of \cite{Ka1}, by Lemma 2.3 of \cite{BS}, for every $n\in\N$ there exists an element $f_n\in\tildeC_n$ such that $P\bigl(\Omega_u\cap\{f_n<1\}\bigr)<P(\Omega_u)/2^{\,n+1}$. Let $A:=\bigcap_{n\in\N}\{f_n\geq 1\}$ and $\xi:=\ind_A$. Then:
$$
P\left(\Omega_u\setminus A\right)
= P\biggl(\,\bigcup_{n\in\N}\bigl(\Omega_u\cap\{f_n<1\}\bigr)\biggr)
\leq \sum_{n\in\N}P\bigl(\Omega_u\cap\{f_n<1\}\bigr)
< \sum_{n\in\N}\frac{P(\Omega_u)}{2^{\,n+1}} = \frac{P(\Omega_u)}{2},
$$
which implies $P(A)>0$, thus showing that $P(\xi\geq 0)=1$ and $P(\xi>0)>0$. Note also that $\xi\leq\ind_Af_n\leq f_n$ $P$-a.s., for every $n\in\N$. Since $f_n\in\conv\bigl\{V_T(x_m,H^m):m\in\N,m\geq n\bigr\}$, for every $n\in\N$, and $x_n\searrow0$ as $n\rightarrow\infty$, this implies that $\xi$ generates an arbitrage of the first kind.
\end{proof}

\begin{Rem}	\label{Herdegen}
We want to mention that the recent paper \cite{Herd} provides an alternative characterisation of NA1 in terms of the equivalent \emph{No Gratis Events (NGE)} condition. In particular, the NGE condition (and, consequently, NA1 as well) is shown to be numéraire-independent. We shall give a very simple proof of the latter property in Corollary \ref{stable}.
\end{Rem}

The following theorem gives several equivalent characterisations of the NA1 condition (another equivalent and useful characterisation will be provided in the next section, see Corollary \ref{NA1-NFLVR}).

\begin{Thm}	\label{NUPBR}
The following are equivalent, using the notation introduced in \eqref{char}-\eqref{char-2} and \eqref{MVT}:
\begin{itemize}
\item[(i)] any (and, consequently, all) of the NA1, NUPBR and NCT conditions holds;
\item[(ii)] $\nu_t=0$ $P\otimes B$-a.e. and $\hatK_T=\int_0^T\!a^{\top}_tc^+_ta_t\,dB_t<\infty$ $P$-a.s., i.e., $\lambda\in L^2_{\text{\upshape{loc}}}(M)$;
\item[(iii)] there exists a tradable martingale deflator;
\item[(iv)] $\D\neq\emptyset$, i.e., there exists a martingale deflator.
\end{itemize}
\end{Thm}
\begin{proof}
$(i)\Rightarrow(ii)$: 
due to Lemma \ref{A1-UPBR-CT}, the NA1, NUPBR and NCT conditions are equivalent. So, let us assume that NUPBR holds. Since NUPBR implies NSA, Theorem \ref{NSA} gives that $\nu_t=0$ $P\otimes B$-a.e. and $\alpha=\inf\bigl\{t\in[0,T]:\hatK_t^{t+h}=\infty,\forall h\in(0,T-t]\bigr\}=\infty$ $P$-a.s. It remains to show that $\hatK_T<\infty$ $P$-a.s. Suppose on the contrary that $P(\tau\leq T)>0$, where $\tau:=\inf\bigl\{t\in[0,T]:\hatK_t=\infty\bigr\}$, so that $P(\hatK_T=\infty)=P(\hatZ_T=0)>0$, where the process $\hatZ$ is defined as in Proposition \ref{weakdfl}. 
Define the sequence $\{\tau_n\}_{n\in\N}$ of stopping times $\tau_n:=\inf\bigl\{t\in[0,T]:\hatK_t\geq n\bigr\}$, for $n\in\N$. Clearly, we have $\tau_n\nearrow\tau$ $P$-a.s. as $n\rightarrow\infty$. 
As shown in equation \eqref{tradability}, we have $\theta^n=\ind_{\lsir0,\tau_n\rsi}\lambda\hatZ^{-1}\in\A_1$ and $G_T(\theta^n)=\hatZ_{T\wedge\tau_n}^{-1}-1$, for every $n\in\N$, so that $G_T(\theta^n)\rightarrow\hatZ_{T\wedge\tau}^{-1}-1$ $P$-a.s. as $n\rightarrow\infty$.
Since $\hatZ_{T\wedge\tau}=0$ on $\{\tau\leq T\}$ and $P(\tau\leq T)>0$, this shows that $\bigl\{G_T(H^n):n\in\N\bigr\}$ cannot be bounded in probability, thus contradicting the assumption that NUPBR holds.	\\
$(ii)\Rightarrow(iii)$:
this follows directly from Proposition \ref{weakdfl}, since $\hatK_T<\infty$ $P$-a.s. implies $\tau=\infty$ $P$-a.s.	\\
$(iii)\Rightarrow(iv)$: by Definition \ref{weak-defl}, this implication is trivial.	\\
$(iv)\Rightarrow(i)$:
let $Z\in\D$ and suppose that the random variable $\xi$ generates an arbitrage of the first kind, so that for every $v\in(0,\infty)$ there exists an element $H^v\in\A_v$ satisfying $V_T(v,H^v)\geq\xi$ $P$-a.s. Due to Lemma \ref{sigma}, the product $Z\,V(v,H^v)=Z\,(v+H^v\cdot S)$ is a non-negative local martingale and, hence, also a supermartingale. As a consequence, for every $v\in(0,\infty)$,
$$
E\bigl[Z_T\,\xi\bigr] \leq E\bigl[Z_TV_T(v,H^v)\bigr] \leq E\bigl[Z_0V_0(v,H^v)\bigr] = v.
$$
Since $Z_T>0$ $P$-a.s., this contradicts the assumption that $P(\xi>0)>0$. Due to Lemma \ref{A1-UPBR-CT}, the NUPBR and NCT conditions hold as well.
\end{proof}

Results analogous to Theorem \ref{NUPBR} have already been obtained in Section 4 of \cite{Ka1}, in Section 3 of \cite{HS} and also earlier in Theorem 2.9 of \cite{CS}. However, the proof given here is rather short and emphasises the role of the tradability of the martingale deflator $\hatZ$ introduced in Proposition \ref{weakdfl}. In particular, it shows that the event $\{\hatK_T=\infty\}$ corresponds to the explosion of the final wealth generated by a sequence of $1$-admissible strategies (see also Section 6 of \cite{CL} for a related discussion).

\begin{Rem}[\emph{The numéraire portfolio}]	\label{numeraire}
The NA1 condition can be shown to be equivalent to the existence of the \emph{numéraire portfolio}, defined as the strictly positive portfolio process $V^*:=V(1,\theta^*)$, $\theta^*\in\A_1$, such that $V(1,\theta)/V^*$ is a supermartingale for all $\theta\in\A_1$ (see e.g. \cite{Be}). In the setting of the present paper, it is easy to verify that, as long as NA1 holds, the numéraire portfolio coincides with the inverse of the tradable martingale deflator $\hatZ$, as follows from Theorem \ref{NUPBR} together with Lemma \ref{sigma} and Fatou's lemma (compare also with \cite{HS}, Lemma 5). The equivalence between NUPBR and the existence of the numéraire portfolio is proved in full generality in \cite{KK} (see also \cite{CL} for related results).
\end{Rem}

An important property of the NA1 condition (as well as of NUPBR and NCT), which is not  shared in general by stronger no-arbitrage conditions (see e.g. \cite{DS95c}), is its invariance with respect to a \emph{change of numéraire}, as shown in the next corollary.

\begin{Cor}	\label{stable}
Let $V:=V(1,\theta)$ be a $P$-a.s. strictly positive portfolio process, for some $\theta\in\A_1$. The NA1 condition holds (for $S$) if and only if the NA1 condition holds for $(S/V,1/V)$.
\end{Cor}
\begin{proof}
Due to Theorem \ref{NUPBR}, it suffices to show that $\D\neq\emptyset$ if and only if there exists a martingale deflator for $(S/V,1/V)$. If $Z\in\D\neq\emptyset$, Lemma \ref{sigma} implies that $Z':=ZV$ is a strictly positive local martingale with $Z'_0=1$. Since $Z'(S/V,1/V)=Z(S,1)\in\Mloc$, this shows that $Z'$ is a martingale deflator for $(S/V,1/V)$. Conversely, if $Z'$ is a martingale deflator for $(S/V,1/V)$ then $Z:=Z'/V$ is a strictly positive local martingale with $Z_0=1$ and $ZS=Z'S/V\in\Mloc$, meaning that $Z\in\D$.
\end{proof}

The next lemma describes the general structure of all martingale deflators. The result goes back to \cite{CS} and \cite{S2} (compare also with \cite{AS0}, Theorem 5), but we give a short proof in the Appendix for the sake of completeness\footnote{Actually, in the one-dimensional case, an analogous result can already be found in \cite{YY} (see also \cite{Jacod}, Theorem 6.11).}.

\begin{Lem}	\label{strdefl}
Suppose that any (and, consequently, all) of the NA1, NUPBR and NCT conditions holds. Then every martingale deflator $Z=(Z_t)_{0\leq t \leq T}$ admits the following representation:
$$
Z = \mathcal{E}(-\lambda\cdot M+N) = \hatZ\,\mathcal{E}(N),
$$
for some $N=(N_t)_{0\leq t \leq T}\in\mMloc$ with $N_0=0$, $\langle N,M\rangle=0$ and $\Delta N>-1$ $P$-a.s. and where the process $\hatZ$ is defined as in Proposition \ref{weakdfl}.
\end{Lem}

Theorem \ref{NUPBR} and Lemma \ref{strdefl} show that $\hatZ=\mathcal{E}(-\lambda\cdot M)$ can be rightfully considered as the \emph{minimal} (tradable) martingale deflator (compare with part 2 of Remark \ref{rem-min} and Remark \ref{weak-MMM}). Indeed, besides being the martingale deflator with the ``simplest'' structure, if $\hatZ$ fails to be a martingale deflator, i.e., if $P(\hatZ_T=0)>0$, then there cannot exist any other martingale deflator. 

\begin{Rem}
The equivalence $(i)\Leftrightarrow(iv)$ in Theorem \ref{NUPBR} has been recently established for general semimartingale models in the papers \cite{Kar2}, in the one-dimensional case, and \cite{T}, in the $\R^d$-valued case (see also \cite{Song} for an alternative proof).
We also want to mention that in \cite{Ka1} it is shown that $\D\neq\emptyset$ is equivalent to the existence of a finitely additive measure $Q$ on $(\Omega,\F)$, weakly equivalent to $P$ and locally countably additive, under which $S$ has a kind of local martingale behavior (see also Section 5 of \cite{Cass} for related results). 
\end{Rem}

\section{No Free Lunch with Vanishing Risk}	\label{S6}

The goal of this section consists in studying the NFLVR condition, on which the classical no-arbitrage pricing theory is based (we refer the reader to \cite{DS} for a complete account thereof), and the relations with the weak no-arbitrage conditions discussed so far. 
As can be seen from Definition \ref{arb}, the NA (and, hence, the NFLVR) condition excludes arbitrage possibilities that may require access to a finite line of credit and, hence, cannot be realized in an unlimited way be every market participant. 
Let us begin this section by introducing another (last) notion of arbitrage.

\begin{Def}	\label{AA}
A sequence $\{H^n\}_{n\in\N}\subset\A_c$, for some $c>0$, generates an \emph{Approximate Arbitrage} if $P(G_T(H^n)\geq 0)\rightarrow1$ as $n\rightarrow\infty$ and there exists a constant $\delta>0$ such that $P(G_T(H^n)>\delta)>\delta$, for all $n\in\N$. If there exists no such sequence we say that the \emph{No Approximate Arbitrage (NAA)} condition holds.
\end{Def}

The notion of approximate arbitrage has been first introduced in \cite{LS} in the context of a complete It\^o-process model and turns out to be equivalent to the notion of free lunch with vanishing risk introduced in Definition \ref{arb}-\emph{(v)}, as shown in the next lemma, the proof of which combines several techniques already employed in \cite{DS94} and \cite{LS}. Recall that $\mathcal{C}=\left(\{G_T(H):H\in\A\}-L^0_+\right)\cap L^{\infty}$, according to the notation introduced at the end of Section \ref{S2}.

\begin{Lem}	\label{approximate}
The NFLVR condition and the NAA condition are equivalent.
\end{Lem}
\begin{proof}
Suppose that NFLVR fails to hold. Then, as in Proposition 3.6 of \cite{DS94}, there exists either an arbitrage opportunity or a cheap thrill. Clearly, if there exists an arbitrage opportunity then there also exists an approximate arbitrage. We now show that the existence of a cheap thrill yields an approximate arbitrage, thus proving that NAA implies NFLVR. Due to Lemma \ref{A1-UPBR-CT} together with Theorem \ref{NUPBR}, the existence of a cheap thrill is equivalent to $0<P(\hatK_T=\infty)=:\delta$. For a fixed $\kappa>1+\delta$ and for every $n\in\N$, define the stopping times
$$
\sigma_n:=\inf\bigl\{t\in[0,T]:\hatZ_t=1/n\bigr\}\wedge T
\qquad\text{and}\qquad 
\varrho_n:=\inf\bigl\{t\in[\sigma_n,T]:\hatZ_t=\hatZ_{\sigma_n}/\kappa\bigr\}\wedge T,
$$
where $\hatZ$ is as in Proposition \ref{weakdfl}. 
For every $n\in\N$, define $Y^n:=\hatZ^{\sigma_n}/\hatZ^{\varrho_n}$ and $H^n:=\ind_{\lsi\sigma_n,\varrho_n\rsi}(\hatZ/\hatZ_{\sigma_n})^{-1}\lambda$.
It\^o's formula implies then the following, for all $t\in[0,T]$:
$$
Y^n_t = 1+\int_0^t\!\frac{\hatZ_{\sigma_n}}{\hatZ_u}\ind_{\{\sigma_n\leq u\leq\varrho_n\}}\lambda_u\,dM_u
+\int_0^t\!\frac{\hatZ_{\sigma_n}}{\hatZ_u}\ind_{\{\sigma_n\leq u\leq\varrho_n\}}\lambda_u^{\top}d\langle M,M\rangle_{\!u}\lambda_u
= 1+G_t(H^n),
$$
thus showing that $\{H^n\}_{n\in\N}\subset\A_1$. Furthermore:
$$
\{\hatK_T=\infty\} = \{\hatZ_T=0\} \subseteq \{\sigma_n<\varrho_n<T\} 
\subseteq \{G_T(H^n)=\hatZ_{\sigma_n}/\hatZ_{\varrho_n}-1\} = \{G_T(H^n)=\kappa-1\}
$$
and, hence, $P(G_T(H^n)\geq \kappa-1)\geq\delta$, $\forall n\in\N$. Since we have $\{G_T(H^n)<0\}\subseteq\{\sigma_n<T\}\cap\{\hatK_T<\infty\}$ and $P(\sigma_n<T,\hatK_T<\infty)\rightarrow 0$ as $n\rightarrow\infty$, we also get
$$
P(G_T(H^n)\geq 0) = 1-P(G_T(H^n)<0) \geq 1-P(\sigma_n<T,\hatK_T<\infty) \rightarrow 1
$$
as $n\rightarrow\infty$, thus showing that the sequence $\{H^n\}_{n\in\N}\subset\A_1$ yields an approximate arbitrage.\\
Conversely, suppose that the sequence $\{H^n\}_{n\in\N}\subset\A_c$ generates an approximate arbitrage. By definition, for every $\varepsilon>0$, we have $P(G_T(H^n)^->\varepsilon)\leq P(G_T(H^n)<0)\rightarrow 0$ as $n\rightarrow\infty$. This means that $G_T(H^n)^-\rightarrow0$ in probability as $n\rightarrow\infty$ and, passing to a subsequence, we can assume that the convergence takes place $P$-a.s. For every $n\in\N$, let $f_n:=G_T(H^n)\wedge\delta\in\mathcal{C}$, so that $P(f_n=\delta)>\delta$ and $f_n^-\rightarrow 0$ $P$-a.s. as $n\rightarrow\infty$. Due to Lemma A1.1 (and the subsequent Remark 2) of \cite{DS94}, there exists a sequence $\{g_n\}_{n\in\N}$, with $g_n\in\conv\{f_n,f_{n+1},\ldots\}$, such that $g_n\rightarrow g$ $P$-a.s. as $n\rightarrow\infty$ for some random variable $g:\Omega\rightarrow[0,\delta]$. 
More precisely, due to the bounded convergence theorem (since $-c\leq g_n\leq\delta$ $P$-a.s. for all $n\in\N$),
$$
\delta\,P(g>0) \geq E[g\ind_{\{g>0\}}] = E[g] = \underset{n\rightarrow\infty}{\lim}E[g_n] \geq \delta^2,
$$
meaning that $\beta:=P(g>0)\geq\delta>0$. Egorov's theorem gives that $g_n$ converges to $g$ as $n\rightarrow\infty$ uniformly on a set $\Omega'$ with $P(\Omega')\geq 1-\beta/2$. For every $n\in\N$, define $h_n:=g_n\wedge\delta\ind_{\Omega'}$, so that $\{h_n\}_{n\in\N}\in\mathcal{C}$ and $h_n\rightarrow g\ind_{\Omega'}$ in the norm topology of $L^{\infty}$, i.e., $g\ind_{\Omega'}\in\overline{\mathcal{C}}\cap L^{\infty}_+$. Since
$$
P(g\ind_{\Omega'}>0) = 1-P\bigl(\{g=0\}\cup\Omega^{'c}\bigr)
\geq P(\Omega')-P(g=0) \geq 1-\beta/2-(1-\beta) = \beta/2>0,
$$
this shows that NFLVR fails to hold, thus proving that NFLVR implies NAA.
\end{proof}

Before formulating the next theorem, which essentially corresponds to the main result of \cite{DS94}, we need to recall the classical and well-known notion of \emph{Equivalent Local Martingale Measure}.

\begin{Def}	\label{ELMM}
A probability measure $Q$ on $(\Omega,\F)$ with $Q\sim P$ is said to be an \emph{Equivalent Local Martingale Measure (ELMM)} for $S$ if $S$ is a local $Q$-martingale.
\end{Def}

\begin{Rem}[\emph{On martingale deflators and ELMMs}]	\label{ELMM-defl}
Suppose that there exists an ELMM $Q$ for $S$. Letting $Z^Q=(Z^Q_t)_{0\leq t \leq T}$ be its density process, Bayes' rule implies that $Z^QS\in\Mloc$, meaning that $Z^Q/Z^Q_0\in\D$. Conversely, in view of Remark \ref{Remdefl}, an element $Z\in\D$ can be taken as the density process of an ELMM if and only if $E[Z_T]=1$.
\end{Rem}

\begin{Thm}	\label{NFLVR}
The following are equivalent, using the notation introduced in \eqref{char}-\eqref{char-2} and \eqref{MVT}:
\begin{itemize}
\item[(i)] the NFLVR condition holds;
\item[(ii)] there exists an ELMM for $S$;
\item[(iii)] $\nu_t=0$ $P\otimes B$-a.e., $\hatK_T=\int_0^T\!a^{\top}_tc^+_ta_t\,dB_t<\infty$ $P$-a.s. and there exists $N=(N_t)_{0\leq t\leq T}\in\mMloc$ with $N_0=0$, $\langle N,M\rangle=0$ and $\Delta N>-1$ $P$-a.s. such that $\hatZ\mathcal{E}(N)\in\M$;
\item[(iv)] the conditions NA1 (or, equivalently, NUPBR/NCT) and NA both hold;
\item[(v)] the NAA condition holds.
\end{itemize}
\end{Thm}
\begin{proof}
\emph{(i)}$\Leftrightarrow$\emph{(ii)}: this is the content of Corollary 1.2 of \cite{DS94}, recalling that $S$ is a continuous (and, hence, locally bounded) semimartingale.\\
\emph{(ii)}$\Leftrightarrow$\emph{(iii)}: this equivalence follows from Theorem \ref{NUPBR} and Lemma \ref{strdefl} together with Remark \ref{ELMM-defl}.\\
\emph{(ii)}$\Rightarrow$\emph{(iv)}: the existence of an ELMM for $S$ implies that $\D\neq\emptyset$. Hence, due to Theorem \ref{NUPBR}, the NA1 condition (as well as NUPBR and NCT) holds. Let $H\in\A$ yield an arbitrage opportunity. Lemma \ref{sigma} and Bayes' rule imply that the process $G(H)$ is a local $Q$-martingale uniformly bounded from below. Due to Fatou's lemma, it is also a $Q$-supermartingale and, hence, $E^Q[G_T(H)]\leq 0$. Since $Q\sim P$, this clearly contradicts the assumption that $H$ yields an arbitrage opportunity.\\
\emph{(iv)}$\Leftrightarrow$\emph{(v)}$\Leftrightarrow$\emph{(i)}: these equivalences follow from Lemma \ref{approximate} together with Proposition 3.6 of \cite{DS94}.
\end{proof}

\begin{Rems}
{\textbf{1)} }
As can be seen from part \emph{(iii)} of Theorem \ref{NFLVR}, the NFLVR condition, unlike the weak no-arbitrage conditions discussed in the previous sections, does not only depend on the characteristics of $S$ but also on the structure of the underlying filtration $\FF$. In particular, this means that in general one cannot prove the existence of arbitrage opportunities by relying on the characteristics of the discounted price process only (to this effect, compare also \cite{KK}, Example 4.7)\footnote{We want to mention that, in some special cases, it is possible to check the NFLVR condition in terms of the characteristics of the discounted price process $S$. For instance, in the case when $S$ is a continuous exponential semimartingale and one can take $dB_t=dt$ in \eqref{char}, a probabilistic characterisation of the absence of arbitrage opportunities in terms of the characteristics of $S$ has been obtained in the recent paper \cite{Ly}. In the case of non-negative one-dimensional Markovian diffusions, necessary and sufficient conditions for the validity of NFLVR are provided in \cite{MU2}.}.

{\textbf{2)} }
We want to warn the reader that NFLVR does not ensure that the measure $\hatQ$ defined by $d\hatQ/dP:=\hatZ_T$ is an ELMM for $S$, since NFLVR fails to imply in general that $E[\hatZ_T]=1$. In view of Remark \ref{weak-MMM}, this amounts to saying that NFLVR does not guarantee the existence of the minimal martingale measure (a counterexample is provided in \cite{DS98a}). In other words, the NFLVR condition cannot be checked by looking only at the properties of the minimal (weak) martingale deflator $\hatZ$, unlike weaker no-arbitrage conditions.

{\textbf{3)} }
There is no general implication between NA1 and NA. On the one hand, as shown in example \ref{example-NFLVR}, it might well be that NA1 holds but nevertheless there exist arbitrage opportunities. On the other hand, it is possible to construct models that admit no arbitrage opportunities but do not satisfy NA1 (an explicit example can be found in Section 4 of \cite{LS}; see also \cite{Hulley}, Example 1.37).
For models based on continuous semimartingales, without assuming a priori the validity of NA1, a characterisation of NA is given in Theorem 9 of \cite{KS} and in Theorem 2.1 of \cite{Str}.
\end{Rems}

The following corollary gives an interesting alternative characterisation of NA1, thus complementing Theorem \ref{NUPBR}.

\begin{Cor}	\label{NA1-NFLVR}
The NA1 condition holds if and only if there exists a $P$-a.s. strictly positive portfolio process $V:=V(1,\theta)$, for some $\theta\in\A_1$, such that the NFLVR condition holds for $(S/V,1/V)$.
\end{Cor}
\begin{proof}
Due to Theorem \ref{NUPBR}, the NA1 condition implies the existence of a tradable martingale deflator $Z$, so that $1/Z=V(1,\theta)$ for some $\theta\in\A_1$ (see also Remark \ref{Remdefl}). By letting $V:=V(1,\theta)$, this means that $1/V\in\Mloc$ and $S/V\in\Mloc$ and so $P$ is an ELMM for $(S/V,1/V)$. Due to Theorem \ref{NFLVR}, this implies that $(S/V,1/V)$ satisfies NFLVR.
Conversely, if NFLVR holds for $(S/V,1/V)$, Theorem \ref{NFLVR} gives the existence of an ELMM $Q$ for $(S/V,1/V)$, with density process $Z^Q$. By Bayes' rule, we have $Z^Q/V\in\Mloc$ and $Z^QS/V\in\Mloc$. This means that $Z:=Z^Q/(Z^Q_0\,V)\in\D$. Theorem \ref{NUPBR} then implies that NA1 holds for $S$.
\end{proof}

In particular, the above corollary shows that, as long as NA1 holds, we can find a suitable numéraire $V$ such that the classical and stronger NFLVR condition holds in the $V$-discounted financial market $(S/V,1/V)$, regardless of the validity of NFLVR for the original financial market. In particular, if NA1 holds, the process $\hatZ$ is a tradable martingale deflator and, hence, letting $\hatV:=1/\hatZ$, the NFLVR condition holds for $(S/\hatV,1/\hatV)$. This suggests that, even in the absence of an ELMM for $S$, the financial market $(S/\hatV,1/\hatV)$ can be regarded as a natural setting for solving pricing and portfolio optimisation problems, as it is indeed proposed in the context of the Benchmark Approach (see e.g. \cite{PH}, Chapter 10).

\section{Examples and counterexamples}	\label{examples}

In the present section, we provide several examples and counterexamples for the different no-arbitrage conditions discussed so far. In particular, we aim at illustrating the relationships \eqref{impl}.

\begin{Ex}
We start by giving an explicit example of a model allowing for increasing profits. Let $N=(N_t)_{0\leq t \leq T}\in\Mloc^c$  and $S:=|N|$. Tanaka's formula (see \cite{RY}, Theorem VI.1.2) gives the following canonical decomposition:
$$
S_t = |N_0|+\int_0^t\sign(N_u)dN_u+L^0_t,
\qquad \text{for all }t\in[0,T],
$$
where the process $L^0=\left(L^0_t\right)_{0\leq t \leq T}$ is the local time of $N$ at the level $0$. Using the notation introduced in Section \ref{S2}, we have $A=L^0$ and $M=\sign(N)\cdot N$. We now show that $dA\ll d\langle M,M\rangle$ does not hold, where $\langle M,M\rangle=\langle N\rangle$.
In fact, Proposition VI.1.3 of \cite{RY} shows that, for almost all $\omega\in\Omega$, the measure (in $t$) $dL^0_t(\omega)$ is carried by the set $\{t:N_t(\omega)=0\}$. However, for almost all $\omega\in\Omega$, the set $\{t:N_t(\omega)=0\}$ has zero measure with respect to $d\langle N\rangle_{\!t}(\omega)$. This means that $L^0$ induces a measure which is singular with respect to the measure induced by $\langle N\rangle$. Theorem \ref{NIP} then implies that NIP fails.

In the present context, we can also explicitly construct a trading strategy generating an increasing profit. For simplicity, let us suppose that $N_0=0$ $P$-a.s. and define the process $H=\left(H_t\right)_{0\leq t \leq T}$ by $H:=\ind_{\{N=0\}\cap\,\lsir0,T\rsi}$. Clearly, $H$ is a bounded predictable process and so $H\in L(S)$. Furthermore, $(H\cdot M)_t=\int_0^tH_u\sign(N_u)\,dN_u=0$ $P$-a.s. for all $t\in[0,T]$.
Note also that $\int\!HdL^0=L^0$, since $dL^0_t(\omega)$ is carried by the set $\{t:N_t(\omega)=0\}$ for almost all $\omega\in\Omega$. Hence:
$$
(H\cdot S)_t
= \int_0^tH_u\sign(N_u)dN_u + \int_0^tH_udL^0_u
= L^0_t = \underset{s\leq t}{\sup}\left(-\int_0^s\sign(N_u)dN_u\right),
\qquad \text{for all }t\in[0,T],
$$
where the last equality follows from Skorohod's lemma (see \cite{RY}, Lemma VI.2.1). This shows that the gains from trading process $G(H)=H\cdot S$ starts from $0$ and is non-decreasing. In particular, $H\in\A_0$. Finally, if we assume that the local martingale $N$ is not identically zero, we also have $P\bigl(G_T(H)>0\bigr)>0$. Indeed, suppose on the contrary that $P\bigl(G_T(H)>0\bigr)=0$, so that $\sup_{s\leq T}\left(-\int_0^s\sign(N_u)dN_u\right)=0$ $P$-a.s. and, hence, $\int_0^t\sign(N_u)dN_u$ $\geq 0$ $P$-a.s. for all $t\in[0,T]$. By Fatou's lemma, this implies that $\sign(N)\cdot N$ is a non-negative supermartingale, being a non-negative continuous local martingale. Since $\left(\sign(N)\cdot N\right)_0=0$, the supermartingale property gives $\sign(N)\cdot N=0$, which in turn implies that $\langle N\rangle=\langle\sign(N)\cdot N\rangle=0$, thus contradicting the assumption that $N$ is not trivial.
\end{Ex}

\begin{Rem}
An interesting interpretation of the arbitrage possibilities arising from local times can be found in \cite{JP}, where it is shown that the existence of large traders (whose orders affect market prices) can introduce ``hidden'' arbitrage opportunities for the small traders, who act as price-takers. These arbitrage profits are ``hidden'' since they occur on time sets of Lebesgue measure zero, being related to the local time of Brownian motion.
Other examples of arbitrage profits arising from local time can be found in \cite{NS} and \cite{Rossello}. Furthermore, in the recent paper \cite{JP2} it is shown that condition \emph{(iii)} of Theorem \ref{NIP} can be violated when projecting an asset price process onto a subfiltration if there is a bubble in the original (larger) filtration.
\end{Rem}

We now present two examples of financial market models that satisfy NIP but allow for strong arbitrage opportunities. In view of Theorem \ref{NSA}, the two following examples satisfy $\nu_t=0$ $P\otimes B$-a.e. but $P(\alpha<T)>0$, meaning that the mean-variance trade-off process is allowed to jump to infinity with positive probability. 

\begin{Ex}		\label{ex-NIP-SA-1}
Let $M=(M_t)_{0\leq t \leq T}\in\Mloc^c$ with $M_0=0$ and let $\tau$ be a stopping time such that $P(\tau<T)>0$. Define the discounted price process $S=(S_t)_{0\leq t \leq T}$ of a single risky asset as follows:
$$
S = M + \langle M\rangle^{\beta}_{\cdot\wedge\tau}
+ \bigl(\langle M\rangle_{\cdot\vee\tau}-\langle M\rangle_{\tau}\bigr)^{\gamma},
$$ 
for some $\gamma\leq 1/2<\beta$. Then, due to It\^o's formula:
$$
dS_t = dM_t +\left(\beta\ind_{\{t\leq\tau\}}\langle M\rangle_t^{\beta-1}
+ \gamma\ind_{\{t>\tau\}}\bigl(\langle M\rangle_t-\langle M\rangle_{\tau}\bigr)^{\gamma-1}\right)d\langle M\rangle_t.
$$
Theorem \ref{NIP} implies that NIP holds. However, on $\{\tau<T\}$ we have that, for every $\varepsilon>0$:
$$
\hatK_{\tau}^{\tau+\varepsilon} = 
\gamma^2\!\!\int_{\tau}^{\tau+\varepsilon}\!\!\bigl(\langle M\rangle_t-\langle M\rangle_{\tau}\bigr)^{2(\gamma-1)}d\langle M\rangle_t
=
\begin{cases}
\gamma^2\log\bigl(\langle M\rangle_{\tau+t}-\langle M\rangle_{\tau}\bigr)\Bigr|_{t=0}^{\varepsilon}	 
&\text{if }\gamma=1/2	\\[0.3cm]
\frac{\gamma^2}{2\gamma-1}\bigl(\langle M\rangle_{\tau+t}-\langle M\rangle_{\tau}\bigr)^{2\gamma-1}\Bigr|_{t=0}^{\varepsilon} 
&\text{if } \gamma<1/2
\end{cases} \;
= \infty.
$$
This shows that in the present example we have $\alpha=\tau$ $P$-a.s. Hence, due to Theorem \ref{NSA}, the NSA condition fails to hold. By letting $M$ be a standard Brownian motion on $(\Omega,\F,\FF,P)$, $\gamma=1/2$ and $\tau=0$, we recover the situation considered in Example 3.4 of \cite{DS95b}.
\end{Ex}

\begin{Ex}		\label{ex-NIP-SA-2}
Let $W=(W_t)_{0\leq t \leq T}$ be a standard Brownian motion on $(\Omega,\F,\FF,P)$ and define $S$ as follows, for all $t\in[0,T]$:
$$
S_t = W_t + \int_0^t\!\frac{W_u}{u}\,du.
$$
Clearly, Theorem \ref{NIP} implies that NIP holds. However, due to Corollary 3.2 of \cite{JY}, we have $\int_0^{\varepsilon}(W_u/u)^2du=\infty$ $P$-a.s. for every $\varepsilon>0$, meaning that $\alpha=0$ $P$-a.s. Theorem \ref{NSA} then shows that NSA fails to hold (compare also with \cite{RSex}, Section 3.4).
\end{Ex}

\begin{Rem}
Strong arbitrage opportunities may also arise when considering insider trading models, where the original filtration $\FF$ is progressively enlarged with an \emph{honest} time $\tau$ which is not an $\FF$-stopping time. More specifically, as shown in \cite{Imk} (see also \cite{FJS}, Section 6), immediate arbitrage opportunities or, equivalently, strong arbitrage opportunities (see Lemma \ref{NIAO-NSA}), can be achieved in the enlarged filtration by trading as soon as $\tau$ occurs.
\end{Rem}

Let us continue by exhibiting a simple model which satisfies NSA but for which NA1 fails to hold (an analogous example can be found in \cite{LW}, Section 3.1). 

\begin{Ex}		\label{example-NA1}
Let $W=(W_t)_{0\leq t \leq T}$ be a standard Brownian motion on $(\Omega,\F,\FF,P)$ and define the process
$X=(X_t)_{0\leq t \leq T}$ as the solution to the following SDE, for some fixed $K>0$:
$$
dX_t = \frac{K-X_t}{T-t}dt+dW_t,
\qquad X_0=0.
$$
The process $X$ is a Brownian bridge (see \cite{RY}, Exercise IX.2.12) starting at the level $0$ and ending at the level $K>0$. Let us define the discounted price process $S=(S_t)_{0\leq t \leq T}$ of a single risky asset as $S_t:=\exp(X_t)$, for $t\in[0,T]$. Then, due to It\^o's formula:
$$
dS_t = S_t\left(\frac{K-\log(S_t)}{T-t}+\frac{1}{2}\right)dt+S_t\,dW_t,
\qquad S_0=1.
$$
It is easy to see that the condition of Theorem \ref{NSA} is satisfied and, hence, there are no strong arbitrage opportunities, since the process $\hatK=\int_0^{\cdot}\bigl(\frac{K-\log\left(S_u\right)}{T-u}+\frac{1}{2}\bigr)^2du$ does not jump to infinity. However, we have $\hatK_t<\infty$ $P$-a.s. for all $t\in\left[0,T\right)$ but $\hatK_T=\infty$ $P$-a.s. (compare also with \cite{C}). Theorem \ref{NUPBR} then implies that NA1 fails to hold.
\end{Ex}

Financial models satisfying NA1 but not NFLVR can typically be found in the context of Stochastic Portfolio Theory (see e.g. \cite{FK}, Sections 5-6) and within the Benchmark Approach (see e.g. \cite{Pl1}, \cite{PH}, Chapters 12-13, and \cite{Hulley}, Chapter 5). Moreover, models satisfying NA1 but not NFLVR can be constructed in a general way by means of absolutely continuous but not equivalent changes of measure (see e.g. \cite{DS95a}, \cite{OR} and, more recently,  \cite{CT}, \cite{Font} and \cite{RR}) and by means of filtration enlargements (see e.g. \cite{FJS}).
We close this section with the following example, which in particular generalizes the classical example based on a three-dimensional Bessel process (see \cite{DS95c}, Corollary 2.10). 

\begin{Ex}	\label{example-NFLVR}
Let $W=(W_t)_{0\leq t\leq T}$ be a standard Brownian motion on the filtered probability space $(\Omega,\F^W_T,\FF^W,P)$, with $\FF^W=(\F^W_t)_{0\leq t \leq T}$ denoting the $P$-augmented natural filtration of $W$, and take a continuous function $\sigma:(0,\infty)\rightarrow(0,\infty)$ such that the following SDE admits a unique strong solution:
\be	\label{SDE-ex}
dS_t = S_t\,\sigma^2(S_t)\,dt+S_t\,\sigma(S_t)\,dW_t,
\qquad S_0=1.
\ee
Assume furthermore that $\int_x^{\infty}\!\frac{1}{y\sigma^2(1/y)}dy<\infty$ for some $x\in(0,\infty)$. According to the notation introduced in Section \ref{S2}, we have $A_t=\int_0^t\!S_u\sigma^2(S_u)du$ and $M_t=\int_0^t\!S_u\sigma(S_u)dW_u$, for all $t\in[0,T]$, and $\lambda=1/S$. Since $S$ is locally bounded and $\sigma(\cdot)$ is continuous, this implies that $\hatK_T=\int_0^T\!\!\lambda^2_td\langle M\rangle_t<\infty$ $P$-a.s., thus showing that NA1 holds (see Theorem \ref{NUPBR}). Since $W$ enjoys the martingale representation property in the filtration $\FF^W$, Lemma \ref{strdefl} implies that $\D=\{\hatZ\}$, where $\hatZ=\mathcal{E}(-\lambda\cdot M)=\mathcal{E}(-\int\!\sigma(S)dW)=1/S$. However, since $\int_x^{\infty}\!\frac{1}{y\sigma^2(1/y)}dy<\infty$, for some $x\in(0,\infty)$, Corollary 4.3 of \cite{MU1} implies that $\hatZ$ is a strict local martingale in the sense of \cite{ELY}, i.e., it is a local martingale which fails to be a true martingale, so that $E[\hatZ_T]<1$. Due to Theorem \ref{NFLVR}, this shows that NFLVR fails for the model \eqref{SDE-ex}.

In the context of the present example, it is easy to construct explicitly an arbitrage opportunity. Indeed, let us define the process $L=(L_t)_{0\leq t \leq T}$ by $L_t:=E[\hatZ_T|\F_t]$, for all $t\in[0,T]$. Then, due to the martingale representation property, there exists an $\FF^W$-predictable process $\theta=(\theta_t)_{0\leq t \leq T}$ with $\int_0^T\!\theta^2_tdt<\infty$ $P$-a.s such that $L=E[\hatZ_T]+\theta\cdot W$ $P$-a.s. Let us also define the process $V:=L/\hatZ=LS$. A simple application of the product rule gives
$$
dV_t = L_tdS_t+S_tdL_t+d\langle L,S\rangle_{\!t}
= L_tS_t\bigl(\sigma^2(S_t)dt+\sigma(S_t)dW_t\bigr)+S_t\theta_tdW_t+S_t\sigma(S_t)\theta_tdt
= \varphi_tdS_t,
$$ 
where the process $\varphi=(\varphi_t)_{0\leq t \leq T}$ is defined as $\varphi_t:=L_t+\theta_t/\sigma(S_t)$, for $t\in[0,T]$. The continuity of $L$, $S$ and of the function $\sigma(\cdot)$ implies that $\varphi\in L(S)$. Noting that $G(\varphi)=V-V_0\geq-E[\hatZ_T]>-1$ $P$-a.s., we also have $\varphi\in\A_1$. Since $G_T(\varphi)=V_T-V_0=1-E[\hatZ_T]>0$ $P$-a.s., this means that $\varphi$ yields an arbitrage opportunity. We have thus shown that the model \eqref{SDE-ex} allows for the possibility of replicating a risk-free zero-coupon bond of unitary nominal value starting from an initial investment which is strictly less than one. However, not every market participant can profit from this arbitrage opportunity in an unlimited way, since the strategy $\varphi\in\A_1$ requires a non-negligible line of credit.

In particular, any function of the form $\sigma(x)=x^{\mu}$, for $\mu<0$, satisfies the integrability condition $\int_x^{\infty}\!\frac{1}{y\sigma^2(1/y)}dy<\infty$ for any $x\in(0,\infty)$. In the special case $\mu=-1$, it can be shown that the process $S$ is a three-dimensional Bessel process (see \cite{RY}, Chapter XI), the classical example of a financial model for which NA (and, hence, NFLVR as well) fails, as shown already in \cite{DS95a}, in Corollary 2.10 of \cite{DS95c} and in Example 4.6 of \cite{KK} (see also\cite{DShi} for related results).
\end{Ex}

\section{Conclusion and discussion}	\label{S7}

In the present paper, we have provided a unified account of several no-arbitrage conditions proposed in the literature in the context of financial market models based on continuous semimartingales. We have focused on the probabilistic characterisations of weak and strong no-arbitrage conditions as well as on the study of their relationships and of their equivalent formulations, illustrating the general theory by means of explicit examples and counterexamples. We now conclude the paper by commenting on the main economic and financial implications of the different no-arbitrage conditions considered so far.

In economic theory, a first and fundamental issue is represented by the relation between no-arbitrage and \emph{market viability}, in the sense of solvability of the portfolio optimisation problem for some hypothetical agent who prefers more to less. In that sense,  a viable model for a financial market is a potential model of a competitive equilibrium. The relations between no-arbitrage conditions and market viability, which go back to the seminal paper \cite{HK}, clarify to what extent the existence of arbitrage profits is incompatible with the possibility of a competitive equilibrium. 

It is easy to see that the minimal NIP requirement does not suffice to ensure any form of market viability. Indeed, Examples \ref{ex-NIP-SA-1} and \ref{ex-NIP-SA-2} show that NIP fails to exclude strong arbitrage opportunities. In the presence of a strong arbitrage opportunity, any agent with non-satiated strictly increasing preferences would invest in it in an unlimited way, because a strong arbitrage opportunity does not require any initial investment nor any amount of credit and, at the same time, yields a positive profit at the final time $T$ (see Definition \ref{arb}-\emph{(ii)}). Of course, such a possibility would contrast with the solvability of portfolio optimisation problems as well as with the existence of an equilibrium, because any agent could always improve the performance of his portfolio at zero cost without risk. 

In view of the preceding discussion, the NSA condition represents a necessary requirement for market viability. In this regard, the result of Theorem 1 of \cite{LW} is of particular interest, since it proves that the NSA condition is actually equivalent to market viability. More specifically, in the context of a complete It\^o process model (and considering utility from intermediate consumption as well), \cite{LW} show that, if NSA holds, then there exists an optimal portfolio for an agent who prefers more to less. However, the agent constructed in \cite{LW} exhibits strictly increasing but rather irregular (and discontinuous) preferences and, most importantly, has no capacity at all for undertaking a net trade requiring access to a credit line. 

It is therefore of interest to study which no-arbitrage condition is equivalent to market viability, defined for a wide class of regular preferences (as it is also the case in the seminal paper \cite{HK}). In this regard, the paper \cite{LW} again provides an interesting result. Indeed, Theorem 2 of \cite{LW} shows that there are no cheap thrills (see Definition \ref{UPBR}), or, equivalently, NA1 holds (see Lemma \ref{A1-UPBR-CT}), if and only if there exists an optimal portfolio for an agent with regular preferences. This result is confirmed and generalised in Proposition 4.19 of \cite{KK}, which shows that, in a general semimartingale setting, the failure of NA1 precludes the solvability of any portfolio optimisation problem (for any strictly increasing concave utility function $U:(0,\infty)\rightarrow\R$). Moreover, in the recent paper \cite{CDM} it has been proved that NA1 is exactly equivalent to the solvability of portfolio optimisation problems (for any strictly increasing concave utility function $U:(0,\infty)\rightarrow\R$ satisfying the Inada conditions and the asymptotic elasticity condition of \cite{KrSch}), up to an equivalent change of measure. Summing up, these results make clear that NA1 can be regarded as the minimal condition in order to ensure a meaningful form of market viability.

Moreover, the fundamental problems of valuation and hedging can be successfully addressed as long as the NA1 condition holds, provided that one replaces ELMMs with martingale deflators (see Definition \ref{weak-defl}). In particular, most of the classical results on the hedging and pricing of contingent claims and on market completeness can also be obtained in terms of martingale deflators, see e.g. \cite{FK}, \cite{FR}, \cite{Ruf} and, in a general semimartingale setting, \cite{BKX} and \cite{SY} (the super-hedging duality can also be extended to financial markets satisfying NA1, see Section 4.7 of \cite{KK}). Furthermore, as shown in Corollary \ref{NA1-NFLVR}, if NA1 holds then we can recover the classical NFLVR condition by means of a change of num\'eraire. This is also related to the Benchmark Approach proposed by Eckhard Platen and collaborators (see e.g. \cite{Pl1}-\cite{PH}), which provides a coherent framework for valuing contingent claims without relying on the existence of risk-neutral measures by considering the num\'eraire portfolio-discounted financial market.

Altogether, the above discussion suggests that the NIP and NSA conditions can be regarded as indispensable ``sanity checks'' and are only meant to exclude almost pathological notions of arbitrage. On the other hand, the NA1 condition, while strictly weaker than the classical NFLVR condition, is equivalent to an economically sound notion of market viability and allows to successfully solve the fundamental problems of portfolio optimisation, pricing and hedging. 

We close the paper with the following table, which summarises the no-arbitrage conditions introduced in Definition \ref{arb} and studied so far, together with their characterisations (see Theorems \ref{NIP}, \ref{NSA}, \ref{NUPBR} and \ref{NFLVR}) and their equivalent formulations (see Lemmata \ref{NIAO-NSA}, \ref{A1-UPBR-CT} and \ref{approximate}).

\begin{center}
\begin{footnotesize}
\begin{tabular}{|p{5.7cm}|p{4.7cm}|p{5.2cm}|}
\hline
\vspace{-0.6cm}\begin{center}\textsc{Condition}\vspace{-0.4cm}\end{center}
& \vspace{-0.6cm}\begin{center}\textsc{Probabilistic Characterisation}\vspace{-0.4cm}\end{center}
& \vspace{-0.6cm}\begin{center}\textsc{Equivalent $\,$formulation}\vspace{-0.4cm}\end{center}	\\
\hline\hline
\emph{No Increasing Profit (NIP)} 
& $\nu_t=0$ $P\otimes B$-a.e.  
& $\qquad\qquad\qquad\quad$-- \\
\hline
\emph{No Strong Arbitrage (NSA)} 
& $\nu_t=0$ $P\otimes B$-a.e. and 

$\alpha=\infty$ $P$-a.s. (i.e., $\hatK$ does not jump to infinity) 
& \emph{No Immediate Arbitrage (NIA)}\\
\hline
\emph{No Arbitrage of the First Kind (NA1)} 
& $\nu_t=0$ $P\otimes B$-a.e. and \newline $\hatK_T<\infty$ $P$-a.s. 
& \emph{No Unbounded Profit with Bounded Risk (NUPBR)} 

\vspace{0.1cm}\emph{No Cheap Thrill (NCT)} \\
\hline
\emph{No Free Lunch with Vanishing Risk (NFLVR)}
& $\nu_t=0$ $P\otimes B$-a.e. and \newline $\hatK_T<\infty$ $P$-a.s. and \newline $\exists N\in\Mloc$ with $N_0=0$, $\langle N,M\rangle=0$, $\Delta N>-1$ $P$-a.s. such that $\hatZ\mathcal{E}(N)\in\M$.
& \emph{No Approximate Arbitrage (NAA)}	\\
\hline
\end{tabular}
\end{footnotesize}
\end{center}

\vspace{1cm}

\noindent \setstretch{1}\begin{footnotesize}\textbf{Acknowledgements:} 
The author is thankful to Monique Jeanblanc, Wolfgang J. Runggaldier and Marek Rutkowski for many useful remarks and discussions on the topic of the present paper. This research was supported by a Marie Curie Intra European Fellowship within the 7th European Community Framework Programme under grant agreement PIEF-GA-2012-332345.\end{footnotesize}\\[-0.7cm]

\appendix
\setstretch{1.2}
\section{Appendix}

\noindent\textbf{Proof of Lemma \ref{sigma}.}\\
The first part of the proof relies on arguments similar to those used in the proofs of Proposition 3.2 of \cite{GLP} and Proposition 8 of \cite{RS}. Let $Z=(Z_t)_{0\leq t \leq T}\in\Dweak$ and $H\in L(S)$. Define the $\R^{d+1}$-valued local martingale $Y=(Y_t)_{0\leq t \leq T}$ by $Y_t:=(Z_tS^1_t,\ldots,Z_tS^d_t,Z_t)^{\top}$  and let $L(Y)$ be the set of all $\R^{d+1}$-valued predictable $Y$-integrable processes, in the sense of Definition III.6.17 of \cite{JS}. For all $n\in\N$, define also $H(n):=H\ind_{\{\|H\|\leq n\}}$. Then, using twice the integration by parts formula and the associativity of the stochastic integral:
$$	\ba
Z\bigl(H(n)\cdot S\bigr)
&= Z_-\cdot\bigl(H(n)\cdot S\bigr) + \bigl(H(n)\cdot S\bigr)_-\cdot Z + \bigl[Z,H(n)\cdot S\bigr]	\\
&= \bigl(Z_-\,H(n)\bigr)\cdot S + \bigl(H(n)\cdot S\bigr)_-\cdot Z + H(n)\cdot \left[S,Z\right]	\\
&= H(n)\cdot(Z_-\cdot S) + \bigl(H(n)\cdot S\bigr)_-\cdot Z + H(n)\cdot \left[S,Z\right]	\\
&= H(n)\cdot(ZS-S_-\cdot Z) + \bigl(H(n)\cdot S\bigr)_-\cdot Z	\\
&= H(n)\cdot\left(ZS\right) + \Bigl(\bigl(H(n)\cdot S\bigr)_--H(n)^{\top}S_-\Bigr)\cdot Z 
= K(n)\cdot Y
\ea	$$
where, for every $n\in\N$, the $\R^{d+1}$-valued predictable process $K(n)$ is defined as $K(n)^i:=H(n)^i$, for all $i=1,\ldots,d$, and $K(n)^{d+1}:=\bigl(H(n)\cdot S\bigr)_--H(n)^{\top}S_-$. Clearly, we have $K(n)\in L(Y)$, since $K(n)$ is predictable and locally bounded, for every $n\in\N$. Define the $\R^{d+1}$-valued predictable process $K$ by $K^i:=H^i$, for all  $i=1,\ldots,d$, and $K^{d+1}:=\left(H\cdot S\right)_--H^{\top}S_-$. Since $H\in L(S)$, $H(n)\cdot S$ converges to $H\cdot S$ in the Emery topology of semimartingales as $n\rightarrow\infty$. This implies that $K(n)\cdot Y=Z\bigl(H(n)\cdot S\bigr)$ also converges in Emery's topology, since the multiplication with $Z$ is a continuous operation. Since the space $\bigl\{K\cdot Y:K\in L(Y)\bigr\}$ is closed in Emery's topology (see \cite{JS}, Proposition III.6.26), we can conclude that $Z(H\cdot S)=\bar{K}\cdot Y$ for some $\bar{K}\in L(Y)$. But since $K(n)$ converges to $K$ ($P$-a.s. uniformly in $t$, at least along a subsequence) as $n\rightarrow\infty$, we can conclude that $\bar{K}=K$ (see \cite{Me}). This shows that $K\in L(Y)$. Since $Y\in\Mloc$ and $K\in L(Y)$, Proposition III.6.42 of \cite{JS} implies that $Z(H\cdot S)=K\cdot Y$ is a $\sigma$-martingale.
To prove the second assertion of the lemma, suppose that $H\in\A$, i.e., there exists a positive constant $a$ such that $(H\cdot S)_t\geq -a$ $P$-a.s. for all $t\in[0,T]$. The process $Z(a+H\cdot S)$ is a $\sigma$-martingale, being the sum of a local martingale and a $\sigma$-martingale. Furthermore, Proposition 3.1 and Corollary 3.1 of \cite{Kall} imply that $Z(a+H\cdot S)$ is a supermartingale, being a non-negative $\sigma$-martingale, and, hence, also a local martingale (compare also with \cite{AS2}, Corollary 3.5). In turn, this implies that $Z(H\cdot S)\in\Mloc$, being the difference of two local martingales.\\

\noindent\textbf{Proof of Lemma \ref{strdefl}.}\\
Let $Z=(Z_t)_{0\leq t \leq T}\in\D$. By Definition \ref{weak-defl} and Remark \ref{Remdefl}, the process $Z\in\Mloc$ satisfies $P\bigl(Z_t>0\text{ and }Z_{t-}>0\text{ for all }t\in[0,T]\bigr)=1$. In view of Theorem II.8.3 of \cite{JS}, the stochastic logarithm $L:=Z^{-1}\cdot Z$ is well-defined as a local martingale with $L_0=0$ and satisfies $Z=\mathcal{E}(L)$. Since $M\in\Mloc^c$, the process $L$ admits a Galtchouk-Kunita-Watanabe decomposition with respect to $M$, see \cite{AS1}. So, we can write $L = \psi\cdot M + N$ for some $\R^d$-valued predictable process $\psi=(\psi_t)_{0\leq t\leq T}\in L^2_{\text{loc}}(M)$ and for some $N=(N_t)_{0\leq t\leq T}\in\Mloc$ with $N_0=0$ and $\langle N,M\rangle=0$. Then, for all $i=1,\ldots,d$:
$$	\ba
ZS^i &= Z_-\cdot S^i + S^i\cdot Z + \langle Z,S^i\rangle
= Z_-\cdot A^i + Z_-\cdot M^i + S^i\cdot Z + \langle Z,M^i\rangle	\\
&= Z_-\cdot\Bigl(\int d\langle M^i,M\rangle\lambda\Bigr) + Z_-\cdot M^i +S^i\cdot Z 
+ Z_-\cdot\bigl\langle\psi\cdot M+N,M^i\bigr\rangle	\\
&= Z_-\cdot\Bigl(\int d\langle M^i,M\rangle\left(\lambda+\psi\right)\Bigr) + Z_-\cdot M^i +S^i_-\cdot Z 
\ea	$$
By Theorem IV.29 of \cite{Pr}, we have $Z_-\cdot M^i\in\Mloc^c$ and $S^i_-\cdot Z\in\Mloc$. In turn, this implies that $Z_-\cdot\bigl(\int d\langle M^i,M\rangle(\lambda+\psi)\bigr)\in\Mloc^c$, for all $i=1,\ldots,d$. Since $Z_->0$ $P$-a.s., Theorem III.15 of \cite{Pr} allows to conclude that $\int\!d\langle M^i,M\rangle(\lambda+\psi)=0$ for all $i=1,\ldots,d$, which in turn implies that the stochastic integral $\psi\cdot M$ is indistinguishable from $-\lambda\cdot M$, thus yielding the following representation:
$$
Z = \mathcal{E}(L) = \mathcal{E}(\psi\cdot M+N)
= \mathcal{E}(-\lambda\cdot M+N)
= \hatZ\,\mathcal{E}\left(N\right)
$$
where the last equality follows by Yor's formula (see e.g. \cite{Pr}, Theorem II.38) and Proposition \ref{weakdfl}. Since $Z>0$ and $\hatZ>0$ $P$-a.s., we also have $\mathcal{E}(N)>0$ $P$-a.s., meaning that $\Delta N>-1$ $P$-a.s.

\end{document}